%% file: RTMwithInfiniteSetsfull.tex
\documentclass{llncs}

\usepackage{amsmath}
\usepackage[usenames]{color}
\usepackage{amssymb}
\usepackage{procalg}
\usepackage{tikz}
\usepackage{amscd}
\usepackage{txfonts}
\usepackage{amsfonts,semantic,colortbl,mathrsfs,stmaryrd,mathtools}
\usepackage{epsfig,graphicx,subfigure}
\usepackage{enumerate,enumitem}
\usepackage{multirow,multicol}
\usepackage{tabularx}
\usepackage{gastex}

\input{macro}

\newcommand{\conf}[1]{}
\newcommand{\arx}[1]{#1}

\title{Reactive Turing Machines with Infinite Alphabets}
\author{Bas Luttik\and Fei Yang}
\institute{Eindhoven University of Technology, The Netherlands}

\begin{document}
\maketitle

\input{Abstract}

\input{Introductions}
\input{Definitions}
\input{RTMInfinite}
\input{Atoms}
\input{RTMLegalOrbitfinite}

\input{Related}
\input{Conclusion}

\bibliographystyle{splncs03}
\bibliography{RTMInfi}

\conf{\input{Appendix}}

\end{document}

%% file: macro.tex
\usepackage{color}
\makeatletter
\@ifundefined{ifAP}{%

\newcommand{\todo}[1]{\underline{\sf todo}: {\color{red}{#1}}}
}{
\ifAP
\else
\fi
\newcommand{\todo}[1]{}
}
\makeatother

\newenvironment{tarray}[2][c]{
  \settowidth{\dimen1}{${} = {}$}%
  \setlength{\arraycolsep}{0.125\dimen1}%
  \hspace{-0.125\dimen1}\array[#1]{#2}\relax
}
{
   \endarray\hspace{-0.125\dimen1}
}
\def\tbb{\begin{tarray}{llll}}
\def\tbbt{\begin{tarray}[t]{llll}}
\def\tee{\end{tarray}}

\newcommand{\T}{\ensuremath{\mathcal{T}} }
\newcommand{\A}{\ensuremath{\mathcal{A}}}
\newcommand{\Atau}{\ensuremath{\mathcal{A}_{\tau}}}
\newcommand{\Sta}{\ensuremath{\mathcal{S}}}
\newcommand{\Atom}{\ensuremath{\mathbb{A}}}
\newcommand{\R}{\mathcal{\mathop{R}}}
\newcommand{\M}{\mathcal{M}}
\newcommand{\D}{\mathcal{D}}
\newcommand{\N}{\mathcal{N}}
\newcommand{\Dbox}{\mathcal{D}_{\Box}}
\newcommand{\tphd}[1]{\check{#1}}
\newcommand{\tphdL}[1]{{#1 \!}^{\scriptscriptstyle <}}
\newcommand{\tphdR}[1]{\prescript{\scriptscriptstyle >}{}{\! #1}}
\newcommand{\Conf}[1][]{\mathcal{C}_{#1}}

\newcommand{\bbisimd}{\ensuremath{\mathrel{\bbisim^{\Delta}}}}
\newcommand{\nil}{\ensuremath{\mathalpha{\mathbf{0}}}}
\newcommand{\encode}[1]{\lceil #1\rceil}

\newcommand{\orb}{\ensuremath{\mathit{orb}}}
\newcommand{\num}{\ensuremath{\mathit{num}}}
\newcommand{\supp}{\ensuremath{\mathit{supp}}}
\newcommand{\atom}{\ensuremath{\mathit{atom}}}
\newcommand{\start}{\ensuremath{\mathit{start}}}
\newcommand{\finish}{\ensuremath{\mathit{finish}}}
\newcommand{\trans}{\ensuremath{\mathit{trans}}}
\newcommand{\cop}{\ensuremath{\mathit{copy}}}
\newcommand{\fresh}{\ensuremath{\mathit{fresh}}}
\newcommand{\refresh}{\ensuremath{\mathit{refresh}}}
\newcommand{\chec}{\ensuremath{\mathit{check}}}
\newcommand{\enu}{\ensuremath{\mathit{enumerate}}}
\newcommand{\gen}{\ensuremath{\mathit{generate}}}
\newcommand{\act}{\ensuremath{\mathit{action}}}
\newcommand{\tran}{\ensuremath{\mathit{transition}}}
\newcommand{\RTMI}{\ensuremath{\textrm{RTM}^\infty}}
\newcommand{\RTMA}{\ensuremath{\textrm{RTM}^\Atom}}

\newcommand{\outcap}[2]{\ensuremath{\mathalpha{\overline{#1}{#2}}}}
\newcommand{\incap}[2]{\ensuremath{\mathalpha{{#1}(#2)}}}
\newcommand{\taucap}{\ensuremath{\mathalpha{\tau}}}
\newcommand{\pref}[1]{\ensuremath{\mathalpha{#1}.}}
\renewcommand{\parc}{\ensuremath{\mathbin{\mid}}}
\newcommand{\restr}[2]{\ensuremath{(#1){#2}}}
\newcommand{\repl}[1]{\ensuremath{\mathalpha{!}#1}}
\newcommand{\aeq}{\ensuremath{\mathrel{=_{\alpha}}}}

\newcommand{\fn}[1]{\ensuremath{\mathsf{fn}(#1)}}
\newcommand{\bn}[1]{\ensuremath{\mathsf{bn}(#1)}}

\newcommand{\piact}[1][\alpha]{\ensuremath{\mathalpha{#1}}}
\newcommand{\inact}[2]{\ensuremath{\mathalpha{{#1}#2}}}
\newcommand{\outact}[2]{\ensuremath{\mathalpha{\overline{#1}{#2}}}}
\newcommand{\boutact}[2]{\ensuremath{\mathalpha{\overline{#1}(#2)}}}
\newcommand{\tauact}{\ensuremath{\mathalpha{\tau}}}
\newcommand{\setb}[1]{\ensuremath{\llbracket{#1}\rrbracket}}
\newcommand{\sete}[1]{\ensuremath{\mathalpha{e_{#1}}}}
\newcommand{\tup}[1]{\ensuremath{\mathalpha{\bar{a}_{#1}}}}

\newcommand{\delete}[1]{} 

%% file: Abstract.tex
\begin{abstract}
The notion of Reactive Turing machine (RTM) was proposed as an orthogonal extension of Turing machines with interaction. RTMs are used to define the notion of executable transition system in the same way as Turing machines are used to define the notion of computable function on natural numbers. RTMs inherited finiteness of all sets involved from Turing machines, and as a consequence, in a single step, an RTM can only communicate elements from a finite set of data. Some process calculi, such as the $\pi$-calculus, essentially depend on an infinite alphabet of actions, and hence it immediately follows that transition systems specified in these calculi are not executable. On closer inspection, however, the $\pi$-calculus does not appear to use the infinite data in a non-computable manner.

In this paper, we investigate several ways to relax the finiteness requirement. We start by considering a variant of RTMs in which all sets are allowed to be countable, and we get a notion of infinitary RTM. Infinitary RTMs are extremely expressive such that we can hardly use them as a expressiveness criterion. Then, we refine the model by adding extra restrictions. As a result, we define a notion of RTM with atoms. It is a more restricted variant of RTMs in which the sets of actions and data symbols are still allowed to be infinite. We propose a notion of of nominal executability based on RTMs with atoms, and show that every effective transition system with atoms is nominally executable. It will follow that processes definable in the $\pi$-calculus are nominally executable. In contrast, in the process specification language mCRL2 it is possible to specify processes that are not nominally executable. Thus, nominal executability provides a new expressiveness criterion for process calculi.
\end{abstract}

%% file: Introductions.tex
\section{Introduction}~\label{sec:intro}

The Turing machine~\cite{Turing1936} is a machine model that formalizes which functions from natural numbers to natural numbers are effectively computable. For a long time, computing functions in a stand-alone fashion was the primary task of computers, but nowadays modern computing systems continuously interact with their environment, and their operations are not supposed to terminate. However, Turing machines lack facilities to adequately deal with the above two important ingredients of modern computing: \emph{interaction} and \emph{non-termination}. In recent decades, quite a number of extended models of computation have been proposed to study the combination of computation and interaction (see, e.g., the collection~\cite{GSW2006}).

The notion of Reactive Turing machine~\cite{BLT2013} was proposed as an orthogonal extension of classical Turing machines with a facility to model interaction in the style of concurrency theory. It subsumes some other concurrent computation models such as the interactive Turing machines~\cite{LY16}.

Reactive Turing machines serve to define which behaviours (labelled transition systems) can be executed by a computing system. We say that a transition system is executable if it is behaviourally equivalent to the transition system of a reactive Turing machine. Note that the notion of executability is parameterised by the choice of a behavioural equivalence: if a behaviour specified in a transition system is not executable up to some fine notion of behavioural equivalence (e.g., divergence-preserving branching bisimilarity), it may still be executable up to some coarser notion of behavioural equivalence (e.g., the divergence-insensitive variant of branching bisimilarity). The entire spectrum of behavioural equivalences~\cite{Glabbeek1993} is at our disposal to draw precise conclusions.

RTMs can be used to characterise the absolute expressiveness of process calculi. In the theory of executability, we ask two interesting questions about a process calculus.
\begin{enumerate}
\item Is it possible to specify every executable behaviour in the process calculus?
\item Is every behaviour specified in the calculus executable?
\end{enumerate}
A one-to-one correspondence was established between the executable behaviours and the behaviours finitely definable (with a guarded recursive specification) in a process calculus with deadlock, a constant denoting successful termination, action prefix, non-deterministic choice, and parallel composition with handshaking communication; the result is up to divergence-preserving branching bisimilarity~\cite{BLT2013}. It was established that every executable behaviour can be specified in the $\pi$-calculus up to divergence-preserving branching bisimilarity~\cite{LY15}.

But it was also observed that $\pi$-calculus processes are generally not executable: The $\pi$-calculus presupposes an infinite set of names, which gives rise to an infinite set of action labels. It is straightforward to define $\pi$-calculus processes that, in fact, execute an unbounded number of distinct actions. The infinity of the set of names is essential in the $\pi$-calculus, both for the mechanism by which input of data is modelled, and for the mechanism by which the notion of private link between processes is modelled. But, one may also argue that these reasons are more syntactic than semantic; the mechanisms themselves are not essentially infinitary. It was already argued that if one abstracts, to some extent, from the two aspects for which the infiniteness is needed, then behaviour defined in the $\pi$-calculus are executable, at least up to branching bisimilarity~\cite{LY15}.

Moreover, a notable number of process calculi leading to transition systems with infinite sets of labels were proposed for various purposes, for instance, the psi-calculus~\cite{bengtson2009psi}, the value-passing calculus~\cite{fu2013value} and mCRL2~\cite{groote_et_al:DSP:2007:862}. We extend the formalism of reactive Turing machines to adapt to the behaviour with infinite sets of labels such as the transition systems of the models mentioned above.

In this paper, we shall first explore a generalised notion of executability based on Reactive Turing Machines that allows an infinite alphabet of actions. First, we shall observe that allowing an infinite alphabet only makes sense if we also allow the set of data symbols (or, equivalently, the set of states) to be infinite. Putting no restrictions at all yields a notion of executability that is not discriminating at all: every countable transition system is executable by an infinitary RTM.\arx{ The result has two immediate corollaries: Every effective transition system is executable modulo divergence-preserving branching bisimilarity by an infinitary RTM with an effective transition relation, and every computable transition system is executable modulo divergence-preserving branching bisimilarity by an infinitary RTM with a computable transition relation.}

Then, we shall consider a more restricted notion of infinitary executability. Following the research about nominal sets for variable binding with infinite alphabets~\cite{GP2002}, the notion of Turing machine with atoms was introduced~\cite{BKLT13}. We define RTMs with atoms as an extension of Turing machines with atoms. RTMs with atoms allow the sets involved in the definition to be infinite, but in a limited way; intuitively, the infinity can only be exploited to generate fresh names in an execution. By using the notion of legal and orbit-finite set, the Turing machines with atoms are allowed to have infinite alphabets, and while keeping the transition relation finitely definable and finite up to atom automorphism. We say a transition system is nominally executable if it is branching bisimilar to a transition system associated with an RTM with atoms.

To characterise the notion of nominal executability, we propose a notion of transition system with atoms as a restricted version of transition systems. We show that every effective transition system with atoms is nominally executable. Moreover, we also conclude that the transition systems associated with an \RTMA{} is an effective transition system with atoms. Therefore, nominally executable transition systems exactly equals to the effective transition systems with atoms.

Finally, we apply the results to draw conclusions about the executability of process calculi. We shall prove that all $\pi$-calculus processes are nominally executable. On the other hand, in mCRL2 it is possible to define behaviours that are not nominally executable.\delete{We also show that behaviours definable in mCRL2 are executable by infinitary RTMs, but the RTMs do not necessarily have an effective transition relation.} Therefore, nominal executability provides a new expressiveness criterion for process calculi involving infinite alphabets.

The paper is organized as follows. In Section~\ref{sec:pre}, the basic definitions of executability are recapitulated, and we also recall some theorems in~\cite{BLT2013,LY15}. In Section~\ref{sec:rtm-infi}, we investigate reactive Turing machines with infinite sets of labels, data symbols and transitions. In Section~\ref{sec:atoms}, we review the notion of set with atoms from~\cite{BLKT11,BKLT13}, propose the reactive Turing machines with atoms, and characterise the class of the transition systems that are nominally executable. In Section~\ref{sec:related work}, the nominal executability of some process calculi involving infinite alphabets is discussed.
The paper concludes in Section~\ref{sec:conclusion}, in which a hierarchy of executable transition systems with infinite sets is proposed. 

\conf{
We provide an Appendix with the definitions from other works, detailed examples and proofs. A full version of this submission is available on arxiv~\cite{LY16a}.
} 

%% file: Definitions.tex
\section{Preliminaries}~\label{sec:pre}

In this section, we briefly recap the theory of executability~\cite{BLT2013}, which is based on RTMs in which all sets involved are finite. We shall generalise the finiteness condition for the sets in later sections.

\paragraph*{The behaviour of discrete-event systems}

We use the notion of transition system to represent the behaviour of discrete-event systems. It is parameterised by a set $\A$ of \emph{action symbols}, denoting the observable events of a system. We shall later impose extra restrictions on $\A$, e.g., requiring that it be finite or have a particular structure, but for now we let $\A$ be just an arbitrary abstract set. We extend $\A$ with a special symbol $\tau$, which intuitively denotes unobservable internal activity of the system. We shall abbreviate $\A \cup\{\tau\}$ by $\Atau$.

\begin{definition}
[Labelled Transition System]\label{def:lts}
An \emph{$\Atau$-labelled transition system} is a triple $(\Sta,\step{},\uparrow)$, where,
\begin{enumerate}
    \item $\Sta$ is a set of \emph{states},
    \item ${\step{}}\subseteq\Sta\times\Atau\times\Sta$ is an $\Atau$-labelled \emph{transition relation} (we write $s\step{a} t$ for$(s,a,t)\in{\step{}}$), and
    \item ${\uparrow}\in\Sta$ is the initial state.
\end{enumerate}
\end{definition}

In this paper, we shall use the notion of (divergence-preserving) branching bisimilarity~\cite{Glabbeek1996,Glabbeek2009}, which is the finest behavioural equivalence in van Glabbeek's linear time - branching time spectrum~\cite{Glabbeek1993}. \conf{(We put the definition in Appendix~\ref{app:bbisim}.) We denote branching bisimilarity by $\bbisim$ and divergence-preserving branching bisimilarity by $\bbisimd$.}
\arx{
In the definition of (divergence-preserving) branching bisimilarity we need the following notation: let $\step{}$ be an $\Atau$-labelled transition relation on a set $\Sta$, and let $a\in\Atau$; we write $s\step{(a)}t$ for ``$s\step{a}t$ or $a=\tau$ and $s=t$''. Furthermore, we denote the transitive closure of $\step{\tau}$ by $\step{}^{+}$ and the reflexive-transitive closure of $\step{\tau}$ by $\step{}^{*}$.

\begin{definition}
[Branching Bisimilarity]\label{def:bbisim}
Let $T_1=(\Sta_1,\step{}_1,\uparrow_1)$ and $T_2=(\Sta_2,\step{}_2,\uparrow_2)$ be $\Atau$-labelled transition systems. A \emph{branching bisimulation} from $T_1$ to $T_2$ is a binary relation $\R\subseteq\Sta_1\times\Sta_2$ such that for all states $s_1$ and $s_2$, $s_1\R s_2$ implies
\begin{enumerate}
    \item if $s_1\step{a}_1s_1'$, then there exist $s_2',s_2''\in\Sta_2$, such that $s_2\step{}_2^{*}s_2''\step{(a)}s_2'$, $s_1\R s_2''$ and $s_1'\R s_2'$;
    \item if $s_2\step{a}_2s_2'$, then there exist $s_1',s_1''\in\Sta_1$, such that $s_1\step{}_1^{*}s_1''\step{(a)}s_1'$, $s_1''\R s_2$ and $s_1'\R s_2'$.
\end{enumerate}
The transition systems $T_1$ and $T_2$ are \emph{branching bisimilar} (notation: $T_1\bbisim T_2$) if there exists a branching bisimulation $\R$ from $T_1$ to $T_2$ s.t. $\uparrow_1\R\uparrow_2$.

A branching bisimulation $\R$ from $T_1$ to $T_2$ is \emph{divergence-preserving} if, for all states $s_1$ and $s_2$, $s_1\R s_2$ implies
\begin{enumerate}
\setcounter{enumi}{2}
    \item if there exists an infinite sequence $(s_{1,i})_{i\in\mathbb{N}}$ such that $s_1=s_{1,0},\,s_{1,i}\step{\tau}s_{1,i+1}$ and $s_{1,i}\R s_2$ for all $i\in\mathbb{N}$, then there exists a state $s_2'$ such that $s_2\step{}^{+}s_2'$ and $s_{1,i}\R s_2'$ for some $i\in\mathbb{N}$; and
    \item if there exists an infinite sequence $(s_{2,i})_{i\in\mathbb{N}}$ such that $s_2=s_{2,0},\,s_{2,i}\step{\tau}s_{2,i+1}$ and $s_1\R s_{2,i}$ for all $i\in\mathbb{N}$, then there exists a state $s_1'$ such that $s_1\step{}^{+}s_1'$ and $s_1'\R s_{2,i}$ for some $i\in\mathbb{N}$.
\end{enumerate}
The transition systems $T_1$ and $T_2$ are \emph{divergence-preserving branching bisimilar} (notation: $T_1\bbisimd T_2$) if there exists a divergence-preserving branching bisimulation $\R$ from $T_1$ to $T_2$ such that $\uparrow_1\R\uparrow_2$.
\end{definition}
}
\paragraph*{A theory of executability}

The notion of reactive Turing machine (RTM)~\cite{BCLT2009,BLT2013} was put forward to mathematically characterise which behaviours are executable by a conventional computing system. In this section, we recall the definition of RTMs and the ensued notion of executable transition system. The definition of RTMs is parameterised with the set $\Atau$, which we  now assume to be a finite set. Furthermore, the definition is parameterised with another finite set $\D$ of \emph{data symbols}. We extend $\D$ with a special symbol $\Box\notin\D$ to denote a blank tape cell, and denote the set $\D\cup\{\Box\}$ of \emph{tape symbols} by $\Dbox$.
\begin{definition}
[Reactive Turing Machine]\label{def:rtm}
A \emph{reactive Turing machine} (RTM) is a triple $(\Sta,\step{},\uparrow)$, where
\begin{enumerate}
    \item $\Sta$ is a finite set of \emph{states},
    \item ${\step{}}\subseteq \Sta\times\Dbox\times\Atau\times\Dbox\times\{L,R\}\times\Sta$ is a finite collection of $(\Dbox\times\Atau\times\Dbox\times\{L,R\})$-labelled \emph{transition relation} (we write $s\step{a[d/e]M}t$ for $(s,d,a,e,M,t)\in{\step{}}$),
    \item ${\uparrow}\in\Sta$ is a distinguished \emph{initial state}.
\end{enumerate}
\end{definition}

Intuitively, the meaning of  a transition $s\step{a[d/e]M}t$ is that whenever the RTM is in state $s$, and $d$ is the symbol currently read by the tape head, then it may execute the action $a$, write symbol $e$ on the tape (replacing $d$), move the read/write head one position to the left or to the right on the tape (depending on whether $M=L$ or $M=R$), and then end up in state $t$. To formalise this intuitive understanding of the operational behaviour of RTMs, we associate with every RTM $\M$ an $\Atau$-labelled transition system  $\T(\M)$. The states of $\T(\M)$ are the
configurations of $\M$, which consist of a state from $\Sta$, its tape contents, and the position of the read/write head.
We denote by $\check{\Dbox}=\{\check{d}\mid d\in\Dbox\}$ the set of \emph{marked} symbols; a \emph{tape instance} is a sequence $\delta\in(\Dbox\cup\check{\Dbox})^{*}$ such that $\delta$ contains exactly one element of the set of marked symbols $\check{\Dbox}$, indicating the position of the read/write head.
We adopt a convention to concisely denote an update of the placement of the tape head marker. Let $\delta$ be an element of $\Dbox^{*}$. Then by $\tphdL{\delta\,}$ we denote the element of $(\Dbox\cup\check{\Dbox})^{*}$ obtained by placing the tape head marker on the right-most symbol of $\delta$ (if it exists), and $\check{\Box}$ otherwise.
Similarly $\tphdR{\,\delta}$ is obtained by placing the tape head marker on the left-most symbol of $\delta$ (if it exists), and $\check{\Box}$ otherwise.

\begin{definition}\label{def:lts-tm}
Let $\M=(\Sta,\step{},\uparrow)$ be an RTM. The \emph{transition system} $\T(\M)$ \emph{associated with} $\M$ is defined as follows:
\begin{enumerate}
\item its set of states is the set $\Conf[\M]=\{(s,\delta)\mid s\in\Sta,\ \text{$\delta$ a tape instance}\}$ of all configurations of $\M$;
    \item its transition relation ${\step{}}\subseteq{\Conf[\M]\times\Atau\times\Conf[\M]}$ is the least relation satisfying, for all $a\in\Atau,\,d,e\in\Dbox$ and $\delta_L,\delta_R\in\Dbox^{*}$:
    \begin{itemize}
        \item $(s,\delta_L\check{d}\delta_R)\step{a}(t,\tphdL{\delta_L}e\delta_R)$ iff $s\step{a[d/e]L}t$, and
        \item $(s,\delta_L\check{d}\delta_R)\step{a}(t,\delta_L e{}\tphdR{\delta_R})$ iff $s\step{a[d/e]R}t$, and
    \end{itemize}
    \item its initial state is the configuration $(\uparrow,\check{\Box})$.
\end{enumerate}
\end{definition}

Turing introduced his machines to define the notion of \emph{effectively computable function}~\cite{Turing1936}. By analogy, the notion of RTM can be used to define a notion of \emph{executable behaviour}. Usually, we shall be interested in the executability modulo (divergence-preserving) branching bisimilarity.

\begin{definition}
[Executability]~\label{def:exe}
A transition system is \emph{executable} modulo (divergence-preserving) branching bisimilarity if it is (divergence-preserving) branching bisimilarity to a transition system associated with some RTM.
\end{definition}

A characterisation of executability modulo (divergence-preserving) branching bisimilarity is given that is independent of the notion of RTM~\cite{GP2002}.

In order to be able to recapitulate some results from our previous work~\cite{BLT2013,LY15}, we need the following definitions, pertaining to the recursive complexity and branching degree of transition systems.
Let $T=(\Sta,\step{},\uparrow)$ be a transition system. We say that $T$ is \emph{effective} if ${\step{}}$ is a recursively enumerable set with respect to some suitable encoding (G\"{o}del numbering). The mapping $\mathalpha{out}:\Sta\rightarrow 2^{\A_{\tau}\times\Sta}$ associates with every state its set of outgoing transitions, i.e., for all $s\in\Sta$, $\mathalpha{out}(s)=\{(a,t)\mid s\step{a}t\}$. We say that $T$ is \emph{computable} if $\mathalpha{out}$ is a recursive function, again with respect to some suitable encoding. We call a transition system \emph{finitely branching} if $\mathalpha{out}(s)$ is finite for every state $s$, and \emph{boundedly branching} if there exists $B\in\mathbb{N}$ such that $|\mathalpha{out}(s)|\leq B$ for all $s\in \Sta$.

The following results were established to characterise the notion of executability~\cite{BLT2013}.
\begin{theorem}\label{thm-blt}
\begin{enumerate}
\item
  For every finite set $\Atau$ and every boundedly branching computable $\Atau$-labelled transition system $T$, there exists an RTM $\M$ such that $T\bbisimd \T(\M)$.
\item
  For every finite set  $\Atau$ and every effective $\Atau$-labelled transition system $T$ there exists an RTM $\M$ such that $T\bbisim \T(\M)$.
\end{enumerate}
\end{theorem}

Moreover, if a transition system without divergence is executable modulo $\bbisimd$, then it is necessarily boundedly branching~\cite{LY15}.
\begin{theorem}~\label{thm-ly15}
If a transition system $T$ has no divergence up to $\bbisimd$ and is unboundedly branching up to $\bbisimd$, then it is not executable modulo $\bbisimd$.
\end{theorem} 

%% file: RTMInfinite.tex
\section{Infinitary Reactive Turing Machines}~\label{sec:rtm-infi}

In this section, we shall investigate the effect of lifting one or more of the finiteness conditions imposed on RTMs on the ensued notion of executability. We start with lifting the finiteness condition on the alphabet of actions and the transition relation only. We shall argue by means of an example that this extension is hardly useful, because it is not possible to associate a different effect with each action. The next step is, therefore, to also allow an infinite set of data symbols. This, in turn, yields a notion of executability that is too expressive. Finally, we provide two intermediate notions of executability by restricting the transition relations associated with infinitary reactive Turing machines to be effective or computable. In this section, we allow $\A$ to be a countably infinite set of action labels.

\paragraph*{Infinitely many states or data symbols}

Recall from Definition~\ref{def:rtm} that an RTM has a finite set of states $\Sta$ and a finite transition relation. If we allow RTMs to have infinitely many actions, then, inevitably, we should at least also allow them to have an infinite transition relation. The following example illustrates that we then also either need infinitely many states or infinitely many data symbols.

\begin{example}~\label{exp:RTM-infi}
Consider an $\Atau$-labelled transition system $T=(\Sta_T,\step{}_T,\uparrow_T)$, where
\begin{enumerate}
    \item ${\Sta_T}=\{\uparrow_T,\downarrow_T\}\cup\{s_x\mid x\in\A\}$, and
    \item ${\step{}_T}=\{(\uparrow_T,x,s_x)\mid x\in\A\}\cup\{(s_x,x,\downarrow_T)\mid x\in\A\}$.
\end{enumerate}
There does not exist an RTM with finitely many states and data symbols that simulates $T$ modulo branching bisimilarity.

Suppose $\M=(\Sta,\step{},\uparrow)$ is an RTM such that $\T(\M)\bbisim T$, and we let $\A=\{x_1,x_2,\ldots\}$.

The transitions $\uparrow_T\step{x_1}s_{x_1},\,\uparrow_T\step{x_2}s_{x_2},\ldots$ lead to infinitely many states $s_{x_1},s_{x_2},\ldots$, which are all mutually distinct modulo branching bisimilarity.

  Let $C=(\uparrow,\tphd{\Box})$ be the initial configuration of $\M$. Assume that we have $C\bbisim \uparrow_T$, so $C$ admits the following transition sequences: $C\step{}^{*}\step{x_1}\step{}^{*}C_1\step{x_1},\,C\step{}^{*}\step{x_2}\step{}^{*}C_2\step{x_2},\ldots$, where $C_1\bbisim s_{x_1},\,C_2\bbisim s_{x_2},\ldots$.

The transitions of an RTM are of the form $(s,a,d,e,M,t)$, where $s,\,t\in\Sta$, and $d,\,e\in\Dbox$; we call the pair $(s,d)$ the trigger of the transition. A configuration $(s',\delta_L\tphd{d'}\delta_R)$ satisfies the trigger $(s,d)$ if $s=s'$ and $d=d'$. Now we observe that a transition $(s,a,d,e,M,t)$ gives
rise to an $a$-transition from every configuration satisfying its trigger $(s,d)$. Since $\Sta$ and $\Dbox$ are finite sets, there are finitely many triggers.

So, in the infinite list of configurations $C_1,C_2,\ldots$, there are at least two configurations $C_i$ and $C_j$, satisfying the same trigger $(s,d)$; these configurations must have the same outgoing transitions.

Now we argue that we cannot have $C_i\bbisim s_{x_i}$ and $C_j\bbisim s_{x_j}$. Since $C_j\step{x_j}$, a transition labelled by $x_j$ is triggered by $(s,d)$. As $C_i$ also satisfies the trigger $(s,d)$, we have the transition $C_i\step{x_j}$. Hence $C_i\not\bbisim s_{x_i}$, and we get a contradiction to $\T(\M)\bbisim T$.
\end{example}

\paragraph*{Infinitary reactive Turing machines}

If we allow the set of control states or the set of data symbols to be infinite too, the expressiveness of RTMs is greatly enhanced. We introduce a notion of infinitary RTM as follows.

\begin{definition}~\label{def:rtmi}
An \emph{infinitary reactive Turing machine} (\RTMI) is a triple $(\Sta,\step{},\uparrow)$, where
\begin{enumerate}
    \item $\Sta$ is a countable set of \emph{states},
    \item ${\step{}}\subseteq \Sta\times\Dbox\times\Atau\times\Dbox\times\{L,R\}\times\Sta$ is a countable collection of $(\Dbox\times\Atau\times\Dbox\times\{L,R\})$-labelled \emph{transition relation} (we write $s\step{a[d/e]M}t$ for $(s,d,a,e,M,t)\in{\step{}}$),
    \item ${\uparrow}\in\Sta$ is a distinguished \emph{initial state}.
\end{enumerate}
\end{definition}

\paragraph*{Executability by an \RTMI}

By analogy to Definition~\ref{def:exe}, we define the executability with respect to \RTMI s.

\begin{definition}~\label{def:rtmi-exe}
A transition system is \emph{executable by an \RTMI{} modulo (divergence-preserving) branching bisimilarity} if it is (divergence-preserving) branching bisimilar to a transition system associated with some \RTMI.
\end{definition}

The following theorem illustrates the expressiveness of \RTMI s, showing that every countable transition system is executable by an \RTMI{} modulo $\bbisimd$. \conf{(We put the proof in Appendix~\ref{app:rtmi}.)}

\begin{theorem}~\label{thm:RTM-infi}
For every countable set $\Atau$ and every countable $\Atau$-labelled transition system $T$, there exists an \RTMI{} $\M$ such that $T\bbisimd \T(\M)$.
\end{theorem}
\arx{
\begin{proof}
 Let $T=(\Sta_T,\step{}_T,\uparrow_T)$ be an $\Atau$-labelled countable transition system, and let $\encode{\_}: \Sta_T\rightarrow \mathbb{N}$ be an injective function encoding its states as natural numbers. Then, an RTM with infinite sets of action symbols and data symbols $\M(T)=(\Sta,\step{},\uparrow)$ is defined as follows.
\begin{enumerate}
    \item $\Sta=\{s,t,\uparrow\}$ is the set of control states.
    \item $\step{}$ is a finite $(\Dbox\times\A\times\Dbox\times\{L,R\})$-labelled \emph{transition relation}, and it consists of the following transitions:
        \begin{enumerate}
            \item $(\uparrow,\tau,\Box,\encode{\uparrow_T},R,s)$,
            \item $(s,\tau,\Box,\Box,L,t)$, and
            \item $(t,a,\encode{s_1},\encode{s_2},R,s)$ for every transition $s_1\step{a}_T s_2$.
        \end{enumerate}
    \item ${\uparrow}\in\Sta$ is the \emph{initial state}.
\end{enumerate}

Note that a transition step $s_1\step{a}s_2$ is simulated by a sequence
\begin{equation*}
(t,\tphd{\encode{s_1}}\Box)\step{a}(s,\encode{s_2}\tphd{\Box})\step{\tau}(t,\tphd{\encode{s_2}}\Box)
\enskip.
\end{equation*}
Then one can verify that $\T(\M(T))\bbisimd T$.
\end{proof}
}

So \RTMI s are very expressive, and they certainly do not yield a useful model to distinguish between processes that can and cannot be executed. Note that the transition relation used to define \RTMI s need not even be computable or effective. As a compromise, we provide two intermediate models by making requirements on the transition relation of \RTMI s.

 We say that a transition relation is effective, if for every pair of a control state and a data symbol $(s,d)$, the set of subsequent transitions is recursively enumerable, i.e., the set $\{(a,e,M,t)\mid s\step{a[d/e]M}t\}$ is recursively enumerable with respect to some encoding. By the proof of Theorem~\ref{thm:RTM-infi}, if the transition system is effective, then the set of transitions $\{(t,a,\encode{s_1},\encode{s_2},R,s)\mid s_1\step{a} s_2\}$ is recursively enumerable. One may trivially verify that all the other transitions are also recursively enumerable. Hence, we get an effective transition relation. We derive the following corollary for the executability of effective transition systems from Theorem~\ref{thm:RTM-infi}.

\begin{corollary}~\label{cor:RTM-infi-effective}
For every countable set $\Atau$ and every effective $\Atau$-labelled transition system $T$, there exists an \RTMI{} $\M$ with an effective transition relation such that $T\bbisimd \T(\M)$.
\end{corollary}

 We say that a transition relation is computable if for every pair of a control state and a data symbol $(s,d)$ the set of subsequent transitions is computable, i.e., the set $\{(a,e,M,t)\mid s\step{a[d/e]M}t\}$ is recursive with respect to some encoding. By analogy to Corollary~\ref{cor:RTM-infi-effective}, we derive the following result.
\delete{
By the proof of Theorem~\ref{thm:RTM-infi}, if the transition system is computable, then the set of transitions $\{(t,a,\encode{s_1},\encode{s_2},R,s)\mid s_1\step{a} s_2\}$ becomes recursive. For the other transitions, they are also trivially recursive. Hence, we get a computable transition relation. Then from Theorem~\ref{thm:RTM-infi} we also derive a corollary for the executability of computable transition systems.}

\begin{corollary}~\label{cor:RTM-infi-computable}
For every countable set $\Atau$ and every computable $\Atau$-labelled transition system $T$, there exists an \RTMI{} $\M$ with a computable transition relation such that $T\bbisimd \T(\M)$.
\end{corollary}

%% file: Atoms.tex
\section{Reactive Turing Machines with Atoms}~\label{sec:atoms}

 In this section, we introduce a notion of reactive Turing machine with atoms (\RTMA) as a natural intermediate between RTMs and \RTMI s. On the one hand, \RTMA s will be more expressive than RTMs, since they will admit infinite alphabets, whereas RTMs do not. On the other hand, \RTMA s will be less expressive than \RTMI s, because there will be restrictions imposed that, intuitively, make the alphabets finitely representable. We introduce a notion of effective transition system with atoms to characterise the transition systems associated with \RTMA s modulo branching bisimilarity. We then have a proper model to investigate the executability of process calculi with infinite alphabets (such as the $\pi$-calculus).

\paragraph*{Sets with atoms}

We adopt the definition of sets with atoms from Boja\'{n}czyk et al.~\cite{BLKT11,BKLT13}. We fix for the remainder of this section a countable infinite set $\Atom$; we call its elements \emph{atoms}. An \emph{atom automorphism} is a bijection (permutation) on $\Atom$. A \emph{set with atoms} is any set that contains atoms or other sets with atoms, in a well-founded way. Every set in the traditional sense thus is a set with atoms. The atoms will allow us to formulate certain finiteness restrictions that are slightly more liberal than simply requiring that sets are finite. To this end we proceed to introduce \emph{legal} and \emph{orbit-finite} sets with atoms.

 For a set with atoms $X$ and an atom automorphism $\pi$, by $\pi(X)$ we denote the set obtained by application of $\pi$ to every atom in $X$, in elements of $X$, in elements of elements of $X$, etc., recursively. For a set of atoms $S\subseteq\Atom$, if an atom automorphism $\pi$ is the identity on $S$, then we call it an \emph{$S$-automorphism}. We say that $S$ supports a set with atoms $X$ if $X=\pi(X)$ for every $S$-automorphism $\pi$. A set with atoms is called \emph{legal} if it has a finite support, each of its elements has a finite support, and so on recursively. A set with atoms may contain infinitely many atoms, but legality restricts the extent. For instance, a finite set is legal (with itself as support), and also a co-finite set is legal (with its finite complement as support). On the other hand, the set of all odd natural numbers is not legal (its support necessarily includes all odd numbers, or all even numbers).

 Now we proceed to introduce the notion of \emph{orbit-finite} set.  Let $x$ be an element in a set with atoms $X$, the \emph{orbit of $x$} is the set \begin{equation*}
 \{y\in X\mid y=\pi(x)\mbox{ for some atom automorphism }\pi\}
 \enskip.
 \end{equation*}
A set with atoms $X$ is partitioned into disjoint \emph{orbits}: elements $x$ and $y$ are in the same orbit iff $\pi(x)=y$ for some atom automorphism $\pi$. For example $\Atom^{2}$ decomposes into two orbits, the diagonal and its complement; and $\Atom^{*}$ has infinitely many orbits as the elements from $\Atom,\,\Atom^{2},\ldots$ all fall into disjoint orbits. A set with atoms that is partitioned into finitely many orbits is called an \emph{orbit-finite} set. Orbit-finiteness restricts the number of partitions of a set with atoms with respect to atom automorphism.

\paragraph*{Reactive Turing machines with atoms}

Boja\'{n}czyk et al.~\cite{BKLT13} defined a notion of Turing machine with atoms based on sets with atoms. Now we generalize this notion by defining a notion of reactive Turing machine with atoms. We assume that the sets of action symbols $\Atau$ and data symbols $\Dbox$ are legal and orbit-finite sets with atoms.

\begin{definition}
[Reactive Turing Machine with atoms]\label{def:rtm-atoms}
A \emph{reactive Turing machine with atoms} (\RTMA) $\M$ is a triple $(\Sta,\step{},\uparrow)$, where
\begin{enumerate}
    \item $\Sta$ is a legal and orbit-finite set of \emph{states},
    \item ${\step{}}\subseteq \Sta\times\Dbox\times\Atau\times\Dbox\times\{L,R\}\times\Sta$ is a legal and orbit-finite $(\Dbox\times\Atau\times\Dbox\times\{L,R\})$-labelled transition relation (we write $s\step{a[d/e]M}t$ for $(s,d,a,e,M,t)\in{\step{}}$),
    \item ${\uparrow}\in\Sta$ is a distinguished \emph{initial state}.
\end{enumerate}
\end{definition}

By analogy to Definition~\ref{def:exe}, we associate with every \RTMA{} a labelled transition system, and define a notion of executability with respect to \RTMA. In this paper, we shall only consider executability modulo branching bisimilarity.

\begin{definition}~\label{executability-RTMA}
A transition system is \emph{nominally executable} if it is branching bisimilar to a transition system associated with some \RTMA.
\end{definition}

\RTMA s give rise to a less liberal notion of executability compared to the one induced by \RTMI s. The following example give us an insight in the effect of the legality restriction.

\begin{example}~\label{example:legality}
Assume that the set of atoms is the set of natural numbers. We consider a transition system only with the following transitions: $s_1\step{1}s_3\step{3}\ldots s_{2n+1}\step{2n+1}\ldots$. It can be simulated by an \RTMI{} with the following transition relation:
\begin{equation*}
\{(s,n,n,n+2,R,t)\mid n=2i+1\,\mbox{and}\, i\in \mathbb{N}\}
\enskip.
\end{equation*}
Either the set of odd numbers or the set of even numbers supports the above set, however, neither of them is finite. The above transition relation therefore cannot define an \RTMA. Moreover, \RTMA s cannot define any transition system with an illegal set of labels, since it is inevitable to introduce an illegal transition relation.
\end{example}

Besides legality, orbit-finiteness also restricts the notion of executability. The transitions are restricted to finitely many different orbits up to atom automorphism. As a result, \RTMA s cannot make transitions labelled with tuples of atoms of arbitrary lengths, nevertheless, such transitions can be realized by an \RTMI.
 \begin{example}~\label{example:orbit-finiteness}
 Consider an \RTMI{} with the following transition relation:
\begin{equation*}
\{(s,\bar{a},\Box,\bar{a},R,t)\mid \bar{a}\mbox{ is a tuple of atoms of arbitrary length}\}
\enskip.
\end{equation*}
The labels of the above transition relation are not orbit-finite, so it does not defines an \RTMA.
\end{example}

%% file: RTMLegalOrbitfinite.tex
\paragraph*{Transition systems with atoms}

Next, we investigate the class of transition systems that are nominally executable. By Example~\ref{example:legality} and Example~\ref{example:orbit-finiteness}, we can easily exclude the transition systems with an illegal or a non-orbit-finite set of labels. We let $\Atau$ be a legal and orbit-finite set of labels for the remaining of this section. We define the notion of transition system with atoms as follows:

\begin{definition}~\label{def:ltsa}
An $\Atau$-labelled transition system $T=(\Sta_T,\step{}_T,\uparrow_T)$ is a \emph{transition system with atoms} if $\Sta_T$ and $\step{}_T$ are legal sets with atoms. We say that a transition system with atoms is $K$-supported if $K\subset\Atom$ and $K$ is a support of the sets $\Sta_T$ and $\step{}_T$.
\end{definition}

We observe that a transition system with atoms $T=(\Sta_T,\step{}_T,\uparrow_T)$ is $K$-supported iff for every $(s,a,t)\in{\step{}_T}$ and for every $K$-automorphism $\pi_K$ we have $\pi_K(s,a,t)\in{\step{}_T}$, where $\pi_K(s,a,t)=(\pi_K(s),\pi_K(a),\pi_K(t))$.

\delete{
\begin{proof}
\begin{enumerate}
\item $\Rightarrow$: We suppose that $K$ is a support of $\step{}_T$, then by the definition of support, for every $K$-automorphism $\pi_K$, we have $\pi_K(\step{}_T)=\step{}_T$. Therefore, for every $(s,a,t)\in{\step{}_T}$, we have $\pi_K(s,a,t)\in{\step{}_T}$.
\item $\Leftarrow$: We suppose that for every $(s,a,t)\in{\step{}_T}$ and for every $K$-automorphism $\pi_K$, we have $\pi_K(s,a,t)\in{\step{}_T}$ then we have $\pi_K(\step{}_T)=\step{}_T$, Therefore $K$ is a support of $\step{}_T$.
\end{enumerate}
\end{proof}
}

For example, the transition systems associated with $\pi$-calculus terms are transition systems with atoms by the structural operational semantics \arx{in Figure~\ref{fig:pi-semantics} }\conf{ in Appendix~\ref{app:pi}}. We consider the set of names as the set of atoms and hence all the $\pi$-terms and transitions are sets with atoms. The support of the union of the transition systems associated with $\alpha$-equivalence class of the individual $\pi$-terms is the empty set. The support of the transition system associated with the $\alpha$-equivalence class of an individual $\pi$-term is the set of free names. Note that the set of free names does not grow by transition~\cite{SW01}.

\paragraph*{Effective transition systems with atoms}

In order to define the class of transition systems that can be simulated by \RTMA s, we need a notion of effectiveness that ensures that the transitions of an effective transition system can be enumerated by an \RTMA. A difficulty we face when defining an RTM with atoms that simulates a legal transition system with atoms
is that \RTMA s are not capable, in general, of decoding an encoded set with atoms. In particular, if the \RTMA{} produces the code of a transition, then it cannot always compute the label of that transition. We illustrate this idea with a simple example.

\begin{example}~\label{example:illegal}
Let $\encode{\_}:\Atom\rightarrow\mathbb{N}$ be an encoding from the set of atoms to natural numbers. Consider the a transition system with a sequence of transitions $\uparrow\step{\encode{x}}s_x\step{x}$ for every $x\in\Atom$, where $\step{\encode{x}}$ denotes a sequence of transitions with the encoding of $x$ as their labels.
There does not exist an \RTMA{} to simulate this transition system.

We suppose that an \RTMA{} $\M$ simulates the above transition system. Note that there are infinitely many $x\in\Atom$ that are not in the support of the $\M$. We choose an $x$ that is not in the support, then for any transition that writes $x$ on the tape or does an $x$-labelled transition from a configuration that does not involve $x$ on the tape, every $\pi(x)$ is created by another transition from that configuration, where $\pi$ is an arbitrary atom automorphism preserving the support of $\M$. Hence, $\M$ is not the required \RTMA.
\end{example}

Rather than encoding the states and transitions of a transition system with atoms, orbits of states and transitions will be encoded. Effectiveness is defined on functions over natural numbers. In order to define the notion of effectiveness, we encode the structures (orbits) of sets with atoms into natural numbers by using the terminology of definable sets introduced by Boja\'{n}czyk~\cite{B2016}.

\begin{definition}~\label{def:sete}
Let $V$ be an countably infinite set of variables. We let $\bar{x}$ be a tuple of variables. An \emph{$\bar{x}$-valuation} is a function that maps each variable in $\bar{x}$ to an element in $\Atom$. The notion of \emph{set builder expression} is inductively defined as follows:
\begin{enumerate}
\item The empty set $\emptyset$ is a set builder expression.
\item A variable is a set builder expression.
\item Let $\bar{x}$ and $\bar{y}$ be disjoint tuples of variables and let $\alpha$ be a set builder expression with free variables contained in $\bar{x}\bar{y}$, and let $\phi$ be a first-order formula over $\Atom$ and with free variables contained in $\bar{x}\bar{y}$. Then $\{\alpha(\bar{x}\bar{y})\mid\mbox{ for }\bar{y} \mbox{ such that }\phi(\bar{x}\bar{y})\}$ is a set builder expression with free variables $\bar{x}$ and bound variables $\bar{y}$. (This expression equals to the set $\{\alpha(\bar{x},\bar{a}) \mid \bar{a}\mbox{ is a tuple of atoms such that }\phi(\bar{x},\bar{a})\}$).
\item If $\alpha_1,\ldots,\alpha_n$ are set builder expressions, then so is $\alpha_1\cup\cdots\cup\alpha_n$.
\end{enumerate}
\end{definition}

We use $\mathbb{B}$ to denote the set of all set builder expressions. For a set builder expression $\alpha$ with free variables $\bar{x}$, we define $\setb{\alpha}$ to be the function which inputs a valuation of $\bar{x}$ and outputs the corresponding set (or set of sets, etc.).

\begin{definition}~\label{def:setb}
A set with atoms is \emph{definable} if it is of the form $\setb{\alpha}(\bar{a})$, where $\alpha$ is a set builder expression and $\bar{a}$ denotes an assignment of atoms to the free variables in $\alpha$.
\end{definition}

We have the following lemma~\cite{B2016}:
\begin{lemma}~\label{lemma:definable}
Every legal set of $n$-tuples of atoms and every orbit-finite set with atoms is definable.
\end{lemma}

Taking the $\pi$-calculus as an example, the set of all $\pi$-terms is not definable, since it is not orbit-finite, in other words, infinitely many distinct structures are involved. However, the $\alpha$-equivalence class of a $\pi$-term is definable since it has only one orbit.

\delete{
 The transition systems associated mCRL2 expression is not necessarily definable. Consider a term that is able to input a string with arbitrary length, this term involves infinitely many distinct structures.}

We introduce the encodings of set builder expressions.  We let $\encode{\_}: \mathbb{B}\rightarrow \mathbb{N}$ be an encoding from set builder expressions to natural numbers. Now we consider an arbitrary legal and orbit-finite set with atoms $x$. We let $\sete{x}$ be a set builder expression of $x$, and $\tup{x}$ be the tuple of atoms satisfying $\setb{\sete{x}}(\tup{x})=x$. In an \RTMA{}, we use a pair of a natural number and a tuple of atoms to represent a set with atoms $x$, i.e., $x$ is represented by $\encode{\sete{x}}$ and $\tup{x}$.

To characterise the class of executable transition systems with atoms, we define a notion of effectiveness on transition systems with atoms on orbits of states and transitions up to atom automorphism.
\delete{
Recall the notion of effective transition system: a transition system is effective, if, for every state, the set of its outgoing transitions is recursively enumerable. For transition systems with atoms, however, sets with atoms cannot be enumerated by \RTMA s directly since the operation of enumerating atoms is not legal. Hence, we define the effectiveness on orbits of states and transitions up to atom automorphism.
}
\begin{definition}
Let $T=(\Sta_T,\step{}_T,\uparrow_T)$ be a transition system with atoms. We say that $T$ is \emph{effective} if it satisfies the following conditions:
\begin{enumerate}
\item Every $s\in\Sta_T$ and every $(s,a,t)\in{\step{}_T}$ is orbit-finite.
\item There exists a recursively enumerable function $\mathit{out}:\mathbb{N}\rightarrow\mathbb{N}$ satisfying:
\begin{equation*}
\forall s\in\Sta_T,\,\mathit{out}(\encode{\sete{s}})=\{\encode{\sete{(s,a,t)}}\mid s\step{a}_T t\}
\enskip.
\end{equation*}
\end{enumerate}
\end{definition}

 Intuitively, a transition system with atoms is effective if there is an effective algorithm to enumerate the set builder expression of the set of outgoing transitions given the set builder expression of a state as the input. Moreover, it is necessary to enumerate the structure of the triple $(s,a,t)$  rather than $(a,t)$, since some variables in $\sete{s}$ might also appear in $\sete{a}$ and $\sete{t}$; the same atoms should be assigned to the variables in $e_s$, $e_a$, and $e_t$ in order to obtain a correct transition.

\paragraph*{Nominal executability}

Now we show that every effective transition system with atoms is nominally executable.

The simulation of a transition from a state $s$ consists of three stages. In the initial stage, we suppose that the representation of the state $\encode{\sete{s}}\tup{s}$ is written on the tape. Then, the \RTMA{} enumerates the structure of an outgoing transition $\encode{\sete{(s,a,t)}}$ on the tape. In the second stage, the \RTMA{} compress from $\encode{\sete{(s,a,t)}}$ and $\tup{s}$ appropriate assignments to the variables in $\tup{a}$ and $\tup{t}$. Fresh atoms are created in this stage if necessary. In the final stage, the \RTMA{} uses $\encode{\sete{a}}$ and $\tup{a}$ to produce an $a$-labelled transition, leading to a configuration representing state $t$, or it returns to the initial stage and enumerates another transition. Note that an \RTMA{} cannot enumerate the outgoing transitions directly, rather, it enumerates the orbit of the transition and nondeterministically produces one of the elements in that orbit.

We use the encodings of set builder expressions to represent the structures of the sets with atoms, moreover, we also need a tuple of atoms to instantiate the free variables. Now we introduce some gadgets of \RTMA s to manipulate tuples of atoms.
\delete{
\begin{lemma}\label{lemma:RTMA-transition-relation}
Let $\M=(\Sta_{\M},\step{}_{\M},\uparrow_{\M})$ be an \RTMA. Then there exists a finite set of atoms $K\subset\Atom$ such that,
for every $(s,a,d,e,M,t)\in{\step{}_{\M}}$, and for every $K\mbox{-automorphism}\,\pi_K$, we have
$\pi_K(s,a,d,e,M,t)\in{\step{}_{\M}}$.
\end{lemma}

\arx{
\begin{proof}
As $\step{}_{\M}$ is a legal set. We take $K$ to be its the minimal support. Then we conclude the lemma from the definition of the support.
\end{proof}
}
}

\begin{example}~\label{example:operation}
Let $\bar{a}$ and $\bar{b}$ be two tuples of atoms. We define an \RTMA{} \conf{(We put the details in Appendix~\ref{app:rtma}.) }$\M$ with $\bar{a}\bar{b}$ as its tape instance, and within finitely many steps of execution, it changes its tape instance by either:
\begin{enumerate}
\item duplicating an atom $x$ from $\bar{a}$, and adding it to $\bar{b}$, or
\item nondeterministically creating a fresh atom $x$ which is not in $\bar{a}$, and adding it to $\bar{b}$.
\end{enumerate}
\arx{
We denote the current tape instance by $\bar{a}\bar{b}$, and we show the two ways to add a new atom to $\bar{b}$.

For the first case, we suppose that $x$ is the atom in $\bar{a}$ to be duplicated, and the first empty cell after $\bar{b}$ is the destination of the duplication. The machine could accomplish the task by the transitions $\cop\step{\tau[x/x]R}\cop_x\step{\tau[y/y]R}{}^{*}\step{\tau[\Box/x]}\finish$, which is realized by the following set of transitions.
\begin{eqnarray*}
&\{(\cop,\tau,x,x,R,\cop_x)\mid x\in\Atom\}\\
&\cup\{(\cop_x,\tau,y,y,R,\cop_x)\mid x,y\in\Atom\}\\
&\cup\{(\cop_x,\tau,\Box,x,R,\finish)\mid x\in\Atom\}
\enskip.
\end{eqnarray*}

This is a legal and orbit-finite set of transitions.

For the second case, the machine creates a fresh atom, by the following set of transitions,
\begin{eqnarray*}
&\{(\fresh,\tau,\Box,x,L,\chec_x)\mid x\in\Atom\}\\
&\cup\{(\chec_x,\tau,y,y,L,\chec_x)\mid x\neq y\wedge y\neq \Box\wedge x,y\in\Atom\}\\
&\cup\{(\chec_x,\tau,x,x,L,\refresh_x)\mid x\in\Atom\}\\
&\cup\{(\refresh_x,\tau,y,y,R,\refresh_x)\mid x\neq y\wedge x,y\in\Atom\}\\
&\cup\{(\refresh_x,\tau,x,y,L,\chec_y)\mid x\neq y\wedge x,y\in\Atom\}\\
&\cup\{(\chec_x,\tau,\Box,\Box,R,\finish)\mid x\in\Atom\}
\enskip.
\end{eqnarray*}

 The machine first creates an arbitrary atom $x$, and then it checks every atom on the tape whether it is identical with $x$. We suppose that $\Box$ indicates the end of the sequence of atoms on tape. If the check procedure succeeds, the creation is finished, otherwise, the machine creates another atom and checks again.
 We also verify that the above transitions form a legal and orbit-finite set.

}
\end{example}

Next, we illustrate that an \RTMA{} is able to produce an $x$-labelled transition from the encoding of the its structure and the atoms used in $x$, if $x$ is from an orbit-finite set $X$.

\begin{example}~\label{example:orb}
Let $X$ be an arbitrary legal and orbit-finite set with atoms. We define an \RTMA{} \conf{(We put the details in Appendix~\ref{app:rtma}.) } such that for any arbitrary $x\in X$, if $\encode{\sete{x}}$ and $\tup{x}$ are written on the tape in a configuration, then there is one and only one labelled transition reachable from that configuration, and the transition is labelled by $x$.

\arx{

We define an \RTMA{} $M=(\Sta_M,\step{}_M,\uparrow_M)$, and we show that $M$ suffices the requirement.

According to the assumption, we suppose that in the state $\start$, the tape instance is $\encode{\sete{x}}\tup{x}$. It suffices to show that within finitely many steps, the machine is able to write $x$ as one symbol on the tape.

Note that $X$ is an orbit-finite set, which means that there are finitely many distinct orbits that construct the set $X$. Therefore, there are finitely many distinct values of $\encode{\sete{x}}$ for all the elements in $X$. The machine associates with each value a program that calculates $\setb{\sete{x}}(\tup{x})$, which produces the elements from that specific orbit according to the valuation of $\bar{x}$. The machine enters the programme by entering the state $\encode{\sete{x}}$.

\begin{equation*}
\{(\start,\tau,\encode{\sete{x}},\encode{\sete{x}},R,\encode{\sete{x}})\mid x\in X\}
\enskip.
\end{equation*}

As $X$ is orbit-finite, there is an upper-bound for the length of the tuple $\bar{x}$ for every $x\in X$. The machine can represent a tuple as a data symbol.
We suppose that every element in $x$-orbit uses $n$ free variables in its structure, then $\bar{x}$ is a tuple of $n$ atoms. Now we consider the following set of transitions:
\begin{equation*}
\{(\encode{\sete{x}},\setb{\sete{x}}(\tup{x}),\tup{x},\setb{\sete{x}}(\tup{x}),R,\finish\mid a_1\ldots a_n\in \Atom\}
\enskip.
\end{equation*}
These transitions will create $x$ by $\setb{\sete{x}}(\tup{x})$. We suppose that $\tup{x}$ is of the form $(a_1,\ldots,a_n)$. The above set of transitions is orbit-finite, since there is an upper bound of $n$,

Moreover, we show that the machine is able to create a tuple $\tup{x}$ as one symbol, given that each atom in $\tup{x}$ is written on one tape cell, and ordered from left to right as the order of the atoms in the tuple.

The machine constructs $\bar{x}$ by duplicating the elements from each tape cell to the tuple one by one, using the transitions as follows:
\begin{eqnarray*}
&\{(\encode{\sete{x}},\tau,a,a,R,(\encode{\sete{x}},a))\mid a\in\Atom\}\\
&\cup\{((\encode{\sete{x}},a),\tau,b,b,R,(\encode{\sete{x}},a))\mid a,b\in\Atom\}\\
&\cup\{((\encode{\sete{x}},a),\tau,\Box,(a),L,(\encode{\sete{x}},a)')\mid a\in\Atom\}\\
&\cup\{((\encode{\sete{x}},a_i),\tau,(a_1,\ldots,a_{i-1}),(a_1,\ldots,a_{i-1},a_i),L,(\encode{\sete{x}},a_i)')\mid a_1,\ldots, a_i\in \Atom, 2 \leq i\leq n\}\\
&\cup\{((\encode{\sete{x}},a)',\tau,b,b,L,(\encode{\sete{x}},a)')\mid a,b\in \Atom,\, a\neq b\}\\
&\cup\{((\encode{\sete{x}},a)',\tau,a,a,R,\encode{\sete{x}})\mid a\in\Atom\}\\
\enskip.
\end{eqnarray*}
The machine first finds the atom to duplicate, and uses a state $(\encode{\sete{x}},a)$ to register the atom $a$. Then it moves the tape head to the tuple and adds the atom to that tuple. If the tuple is empty, then the machine produces $(a)$, otherwise, the machine adds $a$ to an existed tuple $(a_1,\ldots,a_{i-1})$, and enters the state $(\encode{\sete{x}},a)'$. Finally, the tape goes back to the atom it duplicated and enters state $\encode{\sete{x}}$ again to start the duplication of the next atom. This procedure ends by finishing the duplication of all the atoms in $\bar{x}$, and entering the state $\encode{\sete{x}}$ with $\tup{x}$ written under the tape head. Hence, the machine is ready to produce $\setb{\sete{x}}(\tup{x})$.

Since there is an upper bound of the length of the atom, this set of transitions is legal and orbit-finite. Moreover, there are finitely many orbits for the set $X$, which means that the machine needs finitely many such programs.
Hence, we have obtained an \RTMA{} $M$ that meets the requirement.
}
\end{example}

The following lemma shows that effective labelled transition systems with atoms are nominally executable.\conf{ (We put a proof in Appendix~\ref{app:rtma}.)}

\begin{lemma}~\label{lemma:RTMA-effective}
For every legal and orbit-finite set $\Atau$ and every effective $\Atau$-labelled transition system with atoms $T$, there exists an \RTMA{} $\M$ such that $T\bbisim \T(\M)$.
\end{lemma}

\arx{
\begin{proof}

We let $T=(\Sta_T,\step{}_T,\uparrow_T)$ be an effective $\Atau$-labelled transition system with atoms, and let $K\subset\Atom$ be the minimal support of $T$. We show that there exists an \RTMA{} $\M=(\Sta_{\M},\step{}_{\M},\uparrow_{\M})$ such that $\T(\M)\bbisim T$.

As $T$ is effective, for every state $s\in\Sta_T$, the set $\mathalpha{out}(\encode{\sete{s}})=\{\encode{\sete{(s,a,t)}}\mid s\step{a}_T t\}$ is recursively enumerable. We use this fact to simulate the transition system. We describe the simulation in $3$ stages.

\begin{enumerate}
\item Initially, the tape is empty. Hence the initial configuration is $(\uparrow_{\M},\tphd{\Box})$. For simplicity, we do not denote the position of the tape head in the tape instances. The machine first writes the representation of the initial state $\uparrow_T$, i.e., $\encode{\sete{\uparrow_T}}\tup{\uparrow_T}$ on the tape, satisfying $\setb{\sete{\uparrow_T}}(\tup{\uparrow_T})=\uparrow_T$. As $T$ is legal, $\tup{\uparrow_T}$ consists of finitely many atoms. The initialization procedure is represented as follows,

\begin{equation*}
(\uparrow_{\M},\Box)\step{}^{*}(\enu,\encode{\sete{\uparrow_T}}\tup{\uparrow_T})
\enskip.
\end{equation*}

In the control state $\enu$, we assume that the tape instance is $\encode{\sete{s}}\tup{s}$, satisfying $\setb{\sete{s}}(\tup{s})=s$. As the transition system is effective, the machine is able to enumerate the structure of the outgoing transitions of $s$ as follows,

\begin{equation*}
(\enu,\encode{\sete{s}}\tup{s})\step{}^{*}(\gen,\encode{\sete{s}}\tup{s}\encode{\sete{(s,a,t)}})
\enskip.
\end{equation*}

\delete{
For state $s$, its subsequent transitions are $\{(s,a,t)\in \step{}_T\}$ and by the proposition of transition system with atoms, the set of transitions from $\pi_K(s)$ is $\pi_K{\{(s,a,t)\in \step{}_T\}}$. We use this fact to enumerate the transitions from $s$ and take a $\pi_K$ automorphism by the the atoms $\atom(\pi_K(s))$ which is already on the tape, and hence we obtained the transitions from $\pi_K(s)$.

Then the machine enumerates $\mathalpha{out}(s)$,  and writes $(\encode{a},\encode{t})$ (which are represented by $\encode{\orb_{[a]}},\encode{\atom(a)},\encode{\orb_{[t]}},\encode{\atom(t)}$ respectively) on the tape whenever a transition $(s,a,t)$ is enumerated.

The step of enumeration is represented as follows,

\begin{eqnarray*}
&(\enu,\encode{\orb_{[s]}}\encode{\atom(s)}\atom(\pi_K(s)))\step{}^{*}\\
&(\gen,\encode{\orb_{[s]}}\encode{\atom(s)}\encode{\orb_{[a]}}\encode{\atom(a)}\encode{\orb_{[t]}}\encode{\atom(t)}\atom(\pi_K(s)))
\enskip.
\end{eqnarray*}
}
\item In the second stage, the \RTMA{} produces the tuples of atoms $\tup{a}$ and $\tup{t}$ that valuates the free variables of $a$ and $t$. The valuation creates fresh atoms when necessary and preserves atoms from $K$ and $\tup{s}$.

     We denote the tuple of free variables of $\sete{(s,a,t)}$ by $\bar{x}$, and tuples of free variables of $\sete{s}$, $\sete{a}$ and $\sete{t}$ by $\bar{x_s}$, $\bar{x_a}$ and $\bar{x_t}$ respectively. Note that all the variables in $\bar{x_s}$, $\bar{x_a}$ and $\bar{x_t}$ are also in $\bar{x}$. Since $\encode{\sete{(s,a,t)}}$ is already on the tape, the following terms are computable:
    \begin{enumerate}
        \item the set builder expressions of $a$ and $t$: $\encode{\sete{a}}$ and $\encode{\sete{t}}$;
        \item the tuples of free variables: $\bar{x}$, $\bar{x_s}$, $\bar{x_a}$ and $\bar{x_t}$;
    \end{enumerate}
    We show the above statements as follows. We define a triple $(s,a,t)$ by $\{\{s\},\{s,a\},\{s,a,t\}\}$. We use the standard G\"{o}del numbering on sets and variables. The encoding $\encode{\sete{(s,a,t)}}$ is equal to $\encode{\{\{\sete{s}\},\{\sete{s},\sete{a}\},\{\sete{s},\sete{a},\sete{t}\}\}}$. $\encode{\sete{a}}$ and $\encode{\sete{t}}$ are all computable because the projection operation is computable in G\"{o}del encoding of ordered triples. Moreover, the projection from $\encode{\sete{x}}$ to the free variables used in each elements are computable.

   Then we evaluate $\bar{x}$ to tuples of atoms. $\tup{s}$ is a valuation of $\bar{x_s}$. We evaluate $\bar{x}$ by distinguishing two cases:
  \begin{enumerate}
        \item if a variable $y$ in $\bar{x}$ also appears in $\bar{x_s}$, then we duplicate the valuation of that variable from $\tup{s}$ to valuate $y$;
        \item otherwise, we create a fresh atom to valuate $y$.
    \end{enumerate}
    By Example~\ref{example:operation}, the above two operations are valid by \RTMA s.
    Since $\bar{x_a}$ and $\bar{x_t}$ are both sub-tuples of $\bar{x}$, the machine duplicates the valuation from $\bar{x}$ to create $\tup{a}$ and $\tup{t}$.

    Hence, we get $\encode{\sete{a}}$, $\encode{\sete{t}}$, $\tup{a}$ and $\tup{t}$ satisfying that $\setb{\sete{a}}(\tup{a})=a'$, $\setb{\sete{t}}(\tup{t})=t'$, and there exists an $K\cup\tup{s}$-automorphism $\pi$ which preserves all the atoms in $K\cup\tup{s}$, such that $\pi(s,a',t')=(s,a,t)$. By the property of transition system with atoms, $(s,a',t')\in\step{}_T$. Moreover, by Example~\ref{example:operation}, during the generation of fresh variables, every fresh variable from the universe of $\Atom$ can be generated, the \RTMA{} is able to create all the transitions which are equivalent to $(s,a,t)$ up to $K\cup\tup{s}$-automorphism. We denote this stage as follows:

\begin{equation*}
(\gen,\encode{\sete{s}}\tup{s}\encode{\sete{(s,a,t)}})\step{}^{*}(\act,\encode{\sete{s}}\encode{\sete{a}}\encode{\sete{t}}\tup{s}\tup{a}\tup{t})
\enskip.
\end{equation*}

\item In the third stage, the \RTMA{} generates the action label $a'$ and chooses to execute the transition or to continue the enumeration.

As $\Atau$ is a legal and orbit-finite set, by Example~\ref{example:orb}, the \RTMA{} is able to create a symbol $a'$ such that $\setb{\sete{a}}\tup{a}=a'$.

\begin{equation*}
(\act,\encode{\sete{s}}\encode{\sete{a}}\encode{\sete{t}}\tup{s}\tup{a}\tup{t})\step{}^{*}(\tran,\encode{\sete{s}}\encode{\sete{a}}\encode{\sete{t}}\tup{s}\tup{a}\tup{t}a')
\enskip.
\end{equation*}

Then the \RTMA{} has two choices: executing the $a'$-labelled transition and starting the next round of simulation or returning to the first stage.
\begin{eqnarray*}
&(\tran,\encode{\sete{s}}\encode{\sete{a}}\encode{\sete{t}}\tup{s}\tup{a}\tup{t}a')\step{a'}\step{}^{*}(\enu,\encode{\sete{t}}\tup{t})\\
&(\tran,\encode{\sete{s}}\encode{\sete{a}}\encode{\sete{t}}\tup{s}\tup{a}\tup{t}a')\step{}^{*}(\enu,\encode{\sete{s}}\tup{s})\\
\end{eqnarray*}
\end{enumerate}
We can verify that before the $a'$-labelled transition, the transition system of the machine preserves its states modulo $\bbisim$ by a sequence of $\tau$-transitions which leads back to the configuration $(\enu,\encode{\sete{s}}\tup{s})$. Moreover, from the above analysis, we have $(s,a',t')\in\step{}_T$; and every transition obtained by an $K\cup\tup{s}$-automorphism from $(s,a,t)$ can be simulated by the \RTMA{} $M$. Therefore, we conclude that $T\bbisim \T(\M)$.
\end{proof}
}

We also show that the requirements of effective transition systems with atoms are necessary to prove that a transition system is nominally executable.\conf{ (We put a proof in Appendix~\ref{app:rtma}.)}

\begin{lemma}~\label{lemma:RTMA-effective2}
For every \RTMA{} $\M$, the associated transition system $\T(\M)$ is an effective transition system with atoms.
\end{lemma}
\arx{
\begin{proof}
It is obvious that $\T(\M)$ is effective.

Let $\M=(\Sta_{\M},\step{}_{\M},\uparrow_{\M})$, then there exists a finite set of atoms $K\subset\Atom$ such that, for every $(s,a,d,e,M,t)\in{\step{}_{\M}}$, and for every $K\mbox{-automorphism}\,\pi_K$, we have $\pi_K(s,a,d,e,M,t)\in{\step{}_{\M}}$. It follows that the transition system $\T(\M)$ is legal.
\end{proof}
}

To conclude, we have the following theorem stating that the class of nominally executable transition systems are exactly the set of effective transition systems with atoms.
\begin{theorem}~\label{thm:RTMA-effective}
A transition system $T$ is nominally executable iff there exists a legal and orbit-finite set $\Atau$ and an effective $\Atau$-labelled transition system with atoms $T'$ such that $T\bbisim T'$.
\end{theorem}

%% file: Related.tex
\section{Applications}~\label{sec:related work}
\paragraph*{$\pi$-calculus}
\arx{
The $\pi$-calculus was proposed by Milner, Parrow and Walker~\cite{Milner1992} as a language to specify processes with link mobility. In this paper, we shall consider the version presented in the textbook by Sangiorgi and Walker~\cite{SW01}, excluding the match prefix.

We presuppose a countably infinite set $\N$ of names; we use strings of lower case letters for elements of $\N$.
The \emph{prefixes}, \emph{processes} and \emph{summations} of the $\pi$-calculus are, respectively, defined by the following grammar:
\begin{align*}
\pi\      & \coloneqq\ \outcap{x}{y}\ \mid\ \incap{x}{z}\ \mid\ \taucap \qquad (x,y,z\in \N)\\
P\    & \coloneqq\ M\ \mid\  P\parc P\ \mid\ \restr{z}{P}\ \mid\ \repl{P}\\
M\   & \coloneqq\ \nil\ \mid\ \pref{\pi}P \mid\ M \altc M\enskip.
\end{align*}

We use $P\{z/y\}$ to denote a $\pi$-term obtained by substituting every occurrence of $y$ to $z$ in $P$.

An $\alpha$-conversion between $\pi$-terms is defined in~\cite{SW01} as a finite number of renaming of bound names. We write $P\aeq Q$ if $P$ and $Q$ are two $\pi$-terms that are $\alpha$-convertible.

We define the operational behaviour of $\pi$-terms by means of the structural operational semantics in Fig.~\ref{fig:pi-semantics}, in which $\piact{}$ ranges over the set of actions of the $\pi$-calculus.

The transition system associated with a $\pi$-term is defined as follows:

\begin{definition}~\label{lts-piterm}
Let $P$ be a $\pi$-term. $\T(P)=(\Sta_{P},\step{}_{P},\uparrow_{P})$ is the transition system associated with $P$, where
\begin{enumerate}
    \item $\Sta_{P}$ is the set of $\alpha$-equivalence classes of all reachable $\pi$-terms from $P$ by the operational semantics;
    \item $\step{}_{P}$ is the set of transitions between $\alpha$-equivalence classes of all reachable $\pi$-terms; and
    \item $\uparrow_{P}$ is the $\alpha$-equivalence class of $P$.
\end{enumerate}
\end{definition}

\begin{figure}
\begin{center}
\fbox{
\begin{minipage}[t]{0.9\textwidth}
%$\mathrm{STRUCT}\quad\inference{P'\equiv P,\,P\step{\pi}Q,\,Q'\equiv Q}{P'\step{\piact{}}Q'}$\\
$\mathrm{PREFIX}\quad\inference{\,}{\tau.P\step{\tau}P}\quad \inference{}{\overline{x}y.P\step{\overline{x}y}P}\quad\inference{}{x(y).P\step{xz}P\{z/y\}}$\\
$\mathrm{SUM_L}\quad\inference{P\step{\piact{}}P'}{P+Q\step{\piact{}}P'}%\quad\inference{Q\step{\piact{}}Q'}{(P+Q)\step{\piact{}}Q'}$\\
%%$\mathrm{MATCH}\quad\inference{P\step{\alpha}P'}{if\,x=x\,then\,P\step{\alpha}P'}$\\
%%$\mathrm{MISMATCH}\quad\inference{P\step{\alpha}P',\,x\neq y}{if\,x\neq y\,then\,P\step{\alpha}P'}$\\
\quad\mathrm{PAR_L}\quad\inference{P\step{\piact{}}P'}{P\parc{Q}\step{\piact{}}P'\parc{Q}}\,\bn{\piact{}}\cap \fn{Q}=\emptyset$\\
$\mathrm{COM_L}\quad\inference{P\step{\overline{x}y}P',\,Q\step{xy}Q'}{P\parc{Q}\step{\tau}P'\parc{Q'}}\quad
\mathrm{CLOSE_L}\quad\inference{P\step{\overline{x}(z)}P',\,Q\step{xz}Q'}{P\parc{Q}\step{\tau}\restr{z}{(P'\parc{Q'})}}\,z\notin \fn{Q}$\\
$\mathrm{RES}\quad\inference{P\step{\piact{}}P'}{\restr{z}{P}\step{\piact{}}\restr{z}{P'}}\,z\notin\piact{}\quad
\mathrm{OPEN}\quad\inference{P\step{\overline{x}z}P'}{(z)P\step{\overline{x}(z)}P'}\,z\neq x$\\
$\mathrm{REP}
%  \quad\inference{P\mid !P\step{\piact{}}P'}{!P\step{\piact{}}P'}
  \quad\inference{P\step{\piact{}}P'}{\repl{P}\step{\piact{}}P'\parc\repl{P}}
  \quad\inference{P\step{\outact{x}{y}}P',\,P\step{\inact{x}{y}}P''}{\repl{P}\step{\tauact{}}(P'\parc P'')\parc\repl{P}}
  \quad\inference{P\step{\boutact{x}{z}}P',\,P\step{\inact{x}{z}}P''}{\repl{P}\step{\tauact}\restr{z}{(P'\parc P'')}\parc\repl{P}}$\\
$\mathrm{ALPHA}\quad\inference{P\step{\piact{}}P'}{Q\step{\piact{}}P'}\,Q\aeq P$.
\end{minipage}
}\end{center}
\caption{Operational rules for the $\pi$-calculus}\label{fig:pi-semantics}
\end{figure}
}

The motivation for introducing the notion of RTM with infinite alphabets comes from the discussion of the executability of the $\pi$-calculus~\cite{LY15}.
We show that the transition systems that the $\pi$-calculus associates with are nominally executable.
We consider the set of names $\N$ of a $\pi$-calculus process as the set of atoms. The transition system associated with a $\pi$-term \conf{(We put a detailed introduction in Appendix~\ref{app:pi}.) }is actually an effective transition system with atoms. We get the following result as a corollary to Theorem~\ref{thm:RTMA-effective}. \conf{(We prove this result in Appendix~\ref{app:pi}.)}
\delete{
Note that, trivially, $\N$ itself is then a legal and orbit-finite set with atoms. Moreover, a label from the transition system of a $\pi$-calculus process is a pair two names or $\tau$. Hence the set of labels $\Atau$ for every $\pi$-calculus process is a legal and orbit-finite set with atoms. Moreover, the transition relations of all $\pi$-calculus processes are preserved under $\alpha$-conversion. As mentioned in~\cite{SW01}, the set of free variables of a $\pi$-calculus process is invariant under execution. We consider the set of free names in a $\pi$-term as the support of the associated transition system. The set of free names is finite, thus we get a finite support of the transition system. We conclude that the transition systems associated with $\pi$-calculus processes are legal transition systems with atoms. By using the fact that the operational semantics of the $\pi$-calculus leads to an effective transition relation,
}
\begin{corollary}~\label{cor:rtma-pi}
For every $\pi$-calculus process $P$, the transition system $\T(P)$ is nominally executable.
\end{corollary}
\arx{
\begin{proof}
We use $\N$ to denote the countable set of names used in the $\pi$-calculus, and we suppose that $\N=\Atom$. We suppose that $P$ is an arbitrary $\pi$-calculus process. We use $\fn{P}$ to denote the set of free names involved in $P$ and $\bn{P}$ to denote the set of bound names involved in $P$.

By Theorem~\ref{thm:RTMA-effective}, it is sufficient to show that for every $\pi$-calculus process $P$, the transition system $\T(P)$ is an effective legal transition system with atoms. The transition system is effective by the effectiveness of structural operational semantics of the $\pi$-calculus. Therefore, by Definition~\ref{def:ltsa}, it is sufficient to show that there is a finite set $K\subset\N$ such that $K$ is a support of $\T(P)$. We let $\T(P)=(\Sta_{\pi},\step{}_{\pi},P)$ be an $\Atau$ labelled transition system, and we take $K=\fn{P}\cup\{\tau\}$. Note that $\Sta_P$ is the set of $\pi$-terms and $\Atau$ is the set of labels, and hence $\Sta_{\pi}$ is a set with atoms and $\Atau$ is an orbit-finite set with atoms. To show that $K$ is a support of $\T(P)$, we only need to show that for every $(s,a,t)\in\step{}_{\pi}$, and for every $K$-automorphism $\pi_K$, $\pi_K(s,a,t)\in\step{}_{\pi}$.

We let $\pi_K$ be an arbitrary $K$-automorphism, and we show that $\pi_K(P)$ is an $\alpha$-conversion of $P$, i.e., $\pi_K(P)\aeq P$. We show it by a structural induction on $P$.
In the base case, $P=\nil$, then it is trivial that $\pi_K(P)=\nil\aeq P$.

For the step case, we distinguish with 7 cases.
\begin{enumerate}
\item If $P=\pref{\tau}P'$, by induction hypothesis, we have $\pi_K(P')\aeq P'$. Hence, we have $\pi_K(\pref{\tau}P')=\pref{\tau}\pi_K(P')\aeq \pref{\tau}P'=P$.
\item If $P=\pref{\outcap{x}{y}}P'$, by induction hypothesis, we have $\pi_K(P')\aeq P'$. Moreover, $x,y\in\fn{P}\in K$. Hence, we have $\pi_K(\pref{\outcap{x}{y}}P')=\pref{\outcap{x}{y}}\pi_K(P')\aeq \pref{\outcap{x}{y}}P'=P$.
\item If $P=\pref{\incap{x}{y}}P'$, by induction hypothesis, we have $\pi_K(P')\aeq P'$. Moreover, $x\in\fn{P}\in K$. Hence, we have $\pi_K(\pref{\incap{x}{y}}P')=\pref{\incap{x}{\pi_K(y)}}\pi_K(P')\aeq\pref{\incap{x}{y}}P'=P$.
\item If $P=P_1+P_2$, by induction hypothesis, we have $\pi_K(P_1)\aeq P_1,\,\pi_K(P_2)\aeq P_2$. Hence, we have $\pi_K(P_1+P_2)=\pi_K(P_1)+\pi_K(P_2)\aeq P_1+P_2=P$.
\item If $P=P_1\parc P_2$, by induction hypothesis, we have $\pi_K(P_1)\aeq P_1,\,\pi_K(P_2)\aeq P_2$. Hence, we have $\pi_K(P_1\parc P_2)=\pi_K(P_1)\parc \pi_K(P_2)\aeq P_1\parc P_2=P$.
\item If $P=\restr{z}{P'}$, by induction hypothesis, we have $\pi_K(P')\aeq P'$. Moreover, $\pi_K(z)=z'\notin K$. Hence, we have $\pi_K(\restr{z}{P'})=\restr{z'}{\pi_K(P')}\aeq \restr{z}{P'}=P$.
\item If $P=\repl{P'}$, by induction hypothesis, we have $\pi_K(P')\aeq P$. Hence, we have $\pi_K(\repl{P'})=\repl{\pi_K(P')}\aeq \repl{P'}=P$.
\end{enumerate}

By structural induction, we have shown that for every $K$-automorphism $\pi_K$, $\pi_K(P)\aeq P$. Therefore, we have $\pi_K(P)\in\Sta_{\pi}$. Hence, $K$ is a support of $\Sta_{\pi}$.

Next we still let $\pi_K$ be an arbitrary $K$-automorphism, and we show that for every transition $P\step{\piact}_{\pi}Q\in\step{}_{\pi}$, it satisfies that $\pi_K(P\step{\piact}_{\pi}Q)\in\step{}_{\pi}$ by an induction on the structural operational semantics of the $\pi$ calculus.

We construct a proof tree according to the structural operational semantics in Figure~\ref{fig:pi-semantics} for every transition $(P,\piact,Q)\in\step{}_{\pi}$. The induction hypothesis is that, if $(P,\piact,Q)$ is induced from a set of transitions $\mathit{Pre}(P,\piact,Q)\subset\step{}_{\pi}$ , then for every transition $(P_i,\piact_i,Q_i)\in\mathit{Pre}(P,\piact,Q)$, there is $\pi_K(P_i,\piact_i,Q_i)\in\step{}_{\pi}$.

For the base case, the $\nil$ process cannot do any transition, then the property trivially holds. For the step case, we distinguish with several cases as follows.

\begin{enumerate}
\item If the transition is $\pref{\tauact}P\step{\tauact}_{\pi}P$, then we have $\pi_K(\pref{\tauact}P)\step{\pi_K(\tauact)}_{\pi}\pi_K(P)$.
\item If the transition is $\pref{\outcap{x}{y}}P\step{\outact{x}{y}}_{\pi}P$, then we have $\pi_K(\pref{\outcap{x}{y}}P)\step{\pi_K(\outact{x}{y})}_{\pi}\pi_K(P)$.
\item If the transition is $\pref{\incap{x}{y}}P\step{\outact{x}{z}}_{\pi}P\{z/y\}$, then we have $\pi_K(\pref{\incap{x}{y}}P)\step{\pi_K(\inact{x}{z})}_{\pi}\pi_K(P\{z/y\})$.
\item If the transition is $P_1+P_2\step{\piact}_{\pi}Q$, then there are two cases
\begin{enumerate}
    \item if $P_1\step{\piact}_{\pi}Q$. By induction hypothesis, we have $\pi_K(P_1)\step{\pi_K(\piact)}_{\pi}\pi_K(Q)$. Hence, we have $\pi_K(P_1+P_2)\step{\pi_K(\piact)}_{\pi}\pi_K(Q)$.
    \item If $P_2\step{\piact}_{\pi}Q$. The proof is symmetric with the previous case.
\end{enumerate}
\item If the transition is $P_1\parc P_2\step{\piact}_{\pi}Q_1\parc P_2$, then we have $P_1\step{\piact}_{\pi}Q_1$ and $\bn{\piact}\cap\fn{P_2}=\emptyset$. By induction hypothesis, we have $\pi_K(P_1)\step{\pi_K(\piact)}_{\pi}\pi_K(Q_1)$. Moreover, there is $\bn{\pi_K(\piact)}\cap\fn{\pi_K(P_2)}=\emptyset$. Hence we have $\pi_K(P_1\parc P_2)\step{\pi_K(\piact)}_{\pi}\pi_K(Q_1\parc P_2)$.
\item If the transition is $P_1\parc P_2\step{\piact}_{\pi}P_1\parc Q_2$. The  proof is symmetric with the previous case.
\item If the transition is $P_1\parc P_2\step{\tauact}_{\pi}Q_1\parc Q_2$, then we have $P_1\step{\outact{x}{y}}_{\pi}Q_1,\,P_2\step{\inact{x}{y}}_{\pi}Q_2$ (or the symmetrical case). By induction hypothesis, we have $\pi_K(P_1)\step{\pi_K(\outact{x}{y})}_{\pi}\pi_K(Q_1)$ and $\pi_K(P_2)\step{\pi_K(\inact{x}{y})}_{\pi}\pi_K(Q_2)$. Hence we have $\pi_K(P_1\parc P_2)\step{\pi_K(\tauact)}_{\pi}\pi_K(Q_1\parc Q_2)$.
\item If the transition is $P_1\parc P_2\step{\tauact}_{\pi}\restr{z}{(Q_1\parc Q_2)}$, then we have $P_1\step{\boutact{x}{z}}_{\pi}Q_1$, $P_2\step{\inact{x}{z}}_{\pi}Q_2$, and $z\notin\fn{P_2}$ (or the symmetrical case). By induction hypothesis, we have $\pi_K(P_1)\step{\pi_K(\boutact{x}{z})}_{\pi}\pi_K(Q_1)$ and $\pi_K(P_2)\step{\pi_K(\inact{x}{z})}_{\pi}\pi_K(Q_2)$. Moreover, there is $\pi_K(z)\notin\fn{\pi_K(P_2)}$. Hence, we have $\pi_K(P_1\parc P_2)\step{\pi_K(\tauact)}_{\pi}\pi_K(\restr{z}{(Q_1\parc Q_2)})$.

\item If the transition is $\restr{z}{P}\step{\piact}_{\pi}\restr{z}{Q}$, then we have $P\step{\piact}_{\pi} Q$ and $z\notin \piact$. By induction hypothesis, we have $\pi_K(P)\step{\pi_K(\piact)}_{\pi}\pi_K(Q)$. Moreover, there is $\pi_K(z)\notin\pi_K(\piact)$. Hence, we have $\pi_K(\restr{z}{P})\step{\pi_K(\piact)}_{\pi}\pi_K(\restr{z}{Q})$.
\item If the transition is $\restr{z}{P}\step{\boutact{x}{z}}_{\pi}Q$, then we have $P\step{\outact{x}{z}}Q$, and $z\neq x$.  By induction hypothesis, we have $\pi_K(P)\step{\pi_K(\outact{x}{z})}_{\pi}\pi_K(Q)$. Moreover, there is $\pi_K(z)\neq\pi_K(x)$. Hence, we have $\pi_K(\restr{z}{P})\step{\pi_K(\boutact{x}{z})}_{\pi}\pi_K(Q)$.
\item If the transition is $\repl{P}\step{\piact}_{\pi}Q\parc \repl{P}$, then we have $P\step{\piact}_{\pi}Q$. By induction hypothesis, we have $\pi_K(\repl{P})\step{\pi_K(\piact)}_{\pi}\pi_K(Q)$. Hence we have  $\pi_K(P)\step{\pi_K(\piact)}_{\pi}\pi_K(Q\parc \repl{P})$.
\item If the transition is $\repl{P}\step{\tauact}_{\pi}(Q_1\parc Q_2)\parc \repl{P}$, then we have $P\step{\outact{x}{y}}_{\pi}Q_1$, and $P\step{\inact{x}{y}}_{\pi}Q_2$. By induction hypothesis, we have $\pi_K(P)\step{\pi_K(\outact{x}{y})}_{\pi}\pi_K(Q_1)$, and $\pi_K(P)\step{\pi_K(\inact{x}{y})}_{\pi}\pi_K(Q_2)$. Hence we have  $\pi_K(P)\step{\pi_K(\tauact)}_{\pi}\pi_K((Q_1\parc Q_2)\parc \repl{P})$.
\item If the transition is $\repl{P}\step{\tauact}_{\pi}\restr{z}{(Q_1\parc Q_2)}\parc \repl{P}$, then we have $P\step{\boutact{x}{z}}_{\pi}Q_1$, and $P\step{\inact{x}{z}}_{\pi}Q_2$. By induction hypothesis, we have $\pi_K(P)\step{\pi_K(\boutact{x}{z})}_{\pi}\pi_K(Q_1)$, and $\pi_K(P)\step{\pi_K(\inact{x}{z})}_{\pi}\pi_K(Q_2)$. Hence we have  $\pi_K(P)\step{\pi_K(\tauact)}_{\pi}\pi_K(\restr{z}{(Q_1\parc Q_2)}\parc \repl{P})$.

\item If the transition is $P\step{\piact}_{\pi} Q$, and $P\aeq P_1$, then $P_1\step{\piact}_{\pi} Q$. By induction hypothesis $\pi_K(P_1)\step{\pi_K(\piact)}_{\pi} \pi_K(Q)$. Moreover, using the statement $\pi_K(P)\aeq P$, we have $\pi_K(P)\aeq P\aeq P_1\aeq\pi_K(P_1)$. Hence, we have $\pi_K(P)\step{\pi_K(\piact)}_{\pi}\pi_K(Q)$.
\end{enumerate}

Using the induction on the depth of the proof tree of the transition, we have shown that for every transition $P\step{\piact}_{\pi}Q$, we have $\pi_K(P)\step{\pi_K(\piact)}_{\pi}\pi_K(Q)$. We conclude that $K$ is a support of $\step{}_{\pi}$. Therefore, the transition system $\T(P)$ is an effective legal transition system with atoms.

As a consequence of Theorem~\ref{thm:RTMA-effective}, for every $\pi$-calculus process $P$, the transition system $\T(P)$ is nominally executable.
\end{proof}
}
\paragraph*{mCRL2}

The formal specification language mCRL2~\cite{groote_et_al:DSP:2007:862,GM2014} is widely used to specify and analyze the behaviour of distributed systems. The question arises to what extent the transition systems specified by mCRL2 are executable. The actions in an mCRL2 specification may contain a tuple of integers of any arbitrary lengths, which leads to a set of actions with infinitely many orbits. Moreover, we can also specify transition systems that do not have a finite support in mCRL2. Therefore, we conclude that such transition systems are not nominally executable. \conf{(We prove this result in Appendix~\ref{app:mcrl2}.)}

\begin{corollary}~\label{cor:mCRL2}
There exists an mCRL2 specification $P$, such that the transition system $\T(P)$ is not nominally executable.
\end{corollary}
\arx{\begin{proof}
Consider the following mCRL2 specification:
\begin{eqnarray*}
&&\mathit{act\, num:\, Nat};\\
&&\mathit{init\,sum\, v:\, Nat\, .\, num(2 * v)};
\end{eqnarray*}

It defines a transition system that includes a set of transitions from the initial state labelled by all even natural numbers as follows:
\begin{equation*}
\{(\uparrow,2n,\downarrow)\mid n\in\mathbb{N}\}
\enskip.
\end{equation*}

This transition system does not have a finite support, therefore, it is not nominally executable.
\end{proof}
}

\delete{
Moreover, the transition systems specified in mCRL2 are countable but not necessarily effective, therefore, we have the following corollary.
\begin{corollary}~\label{cor:rtmi-mCRL2}
\begin{enumerate}
\item For every mCRL2 specification $P$, the transition system $\T(P)$ is executable by an \RTMI{} modulo $\bbisimd$.
\item There exists an mCRL2 specification $P$, such that the transition system $\T(P)$ is not executable by any \RTMI{} with an effective/computable transition relation modulo $\bbisimd$.
\end{enumerate}
\end{corollary}
}

%% file: Conclusion.tex
\section{Conclusions}~\label{sec:conclusion}

In this paper, we investigated the executable transition systems associated with \RTMI s and \RTMA s.
We summarize the executable transition systems associated with \RTMI s, \RTMA s and RTMs as a hierarchy in Figure~\ref{fig:hierarchy}. Note that all the inclusion relations over the above sets of transition systems are interpreted as the inclusion of the regions in Figure~\ref{fig:hierarchy}.

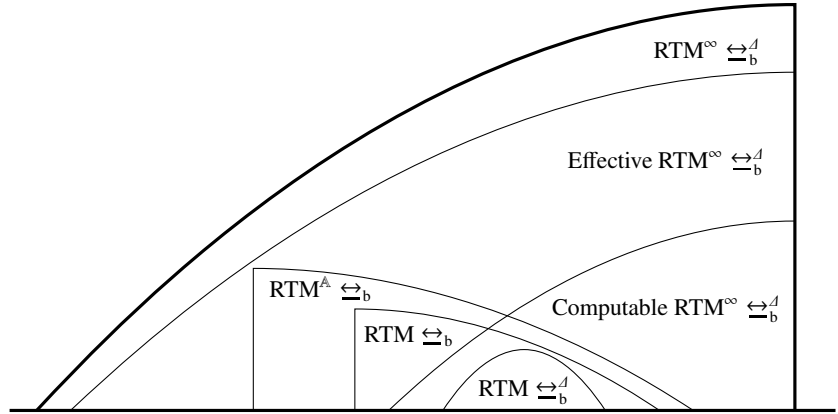
\begin{figure}
\centering
\begin{tikzpicture}
\pgftransformscale{0.9}
\draw[very thick] (6.6,0) -- (-5.6,0);

\draw (3.2,0) parabola bend (2,0.9) (0.8,0);
\node at (2,0.3) {RTM $\bbisimd$};

\draw (4,0) parabola bend (-0.5,1.5) (-0.5,0);
\node at (0.3,1.1) {RTM $\bbisim$};

\draw (-2,0) parabola bend (-2,2.1) (4.5,0);
\node at (-1,1.8) {\RTMA{} $\bbisim$};

\draw (0,0) parabola bend (6,2.8) (6,0);
\node at (4.1,1.5) {Computable \RTMI{} $\bbisimd$};

\draw (-4.7,0) parabola bend (6,5) (6,0);
\node[anchor=north] at (4.1,4) {Effective \RTMI{} $\bbisimd$};

\draw[very thick] (-5.2,0) parabola bend (6,6) (6,0);
\node at (4.7,5.3) {\RTMI{} $\bbisimd$};

\end{tikzpicture}
\caption{A Hierarchy of Executability}~\label{fig:hierarchy}
\end{figure}

We summarize the corresponding transition system of each notion of executability as follows:

\begin{enumerate}
\item by Theorem~\ref{thm-blt}, the class of executable transition systems by RTMs modulo $\bbisimd$ is the boundedly branching computable transition system with a finite set of labels;
\item the class of executable transition systems by RTMs modulo $\bbisim$ is the effective transition system with a finite set of labels;
\item by Theorem~\ref{thm:RTMA-effective}, the class of nominally executable transition systems is the effective transition system with atoms;
\item by Corollary~\ref{cor:RTM-infi-computable}, the class of executable transition systems by \RTMI s with a computable transition relation modulo $\bbisimd$ is the computable transition system;
\item by Corollary~\ref{cor:RTM-infi-effective},the class of executable transition systems by \RTMI s with an effective transition relation modulo $\bbisimd$ is the effective transition system; and
\item by Theorem~\ref{thm:RTM-infi}, the class of executable transition systems by \RTMI s modulo $\bbisimd$ is the countable transition system.
\end{enumerate}

Finally, we propose some future work on this issue.

\begin{enumerate}
\item The precise characterisation of the transition systems executable by \RTMA s modulo $\bbisimd$ is still open. Further restrictions should be imposed to make it possible to generate all possible transitions of an arbitrary state in the transition system from a single configuration of an \RTMA.

\item It would be interesting to show the existence of a universal \RTMA, such that it is able to simulate the behaviour of every \RTMA{} with its encoding.

\item Psi-calculi~\cite{bengtson2009psi} were introduced as to characterise transition systems with nominal data types for data structures and with logical assertions representing facts about data. The adoption of nominal data types provides a natural characterisation of the behaviour executed by \RTMA.\arx{ An encoding of the $\pi$-calculus was proposed in the psi-calculus~\cite{bengtson2011psi}, proving that the psi-calculus is at least as expressive as the $\pi$-calculus.} It would be interesting to figure out the relationship between the transition systems associated with the psi-calculus and the transition systems associated with \RTMA. We conjecture that as long as the logical assertions used in psi-calculus are semi-decidable, the transition systems of the psi-calculus processes are nominally executable.

\item A notion of nominal transition system was proposed by Parrow et al.~\cite{DBLP:conf/concur/ParrowBEGW15}. Nominal transition systems satisfy the requirements of transition systems with atoms naturally. We did not use the notion of nominal transition system since the predicates for states in Hennessy-Milner logic is ignored in proving the executability. By assuming a appropriate definition of effectiveness, we conjecture that the effective nominal transition systems are nominally executable.

\item
The value-passing calculus~\cite{fu2013value} is a process calculus in which the contents of communications are values chosen from natural numbers. It can be used to specify transition systems that are not nominally executable (such as the one we used in the proof of Corollary~\ref{cor:mCRL2}). It would be interesting to impose some structures on the sets of atoms, e.g., by considering the natural numbers. Moreover, we could investigate a notion of RTM with natural numbers, and make a comparison to the value-passing calculus.
\end{enumerate}

%% file: Appendix.tex
\newpage
\appendix

\section{Branching bisimilarity}~\label{app:bbisim}
In the definition of (divergence-preserving) branching bisimilarity we need the following notation: let $\step{}$ be an $\Atau$-labelled transition relation on a set $\Sta$, and let $a\in\Atau$; we write $s\step{(a)}t$ for ``$s\step{a}t$ or $a=\tau$ and $s=t$''. Furthermore, we denote the transitive closure of $\step{\tau}$ by $\step{}^{+}$ and the reflexive-transitive closure of $\step{\tau}$ by $\step{}^{*}$.

\begin{definition}
[Branching Bisimilarity]\label{def:bbisim}
Let $T_1=(\Sta_1,\step{}_1,\uparrow_1)$ and $T_2=(\Sta_2,\step{}_2,\uparrow_2)$ be $\Atau$-labelled transition systems. A \emph{branching bisimulation} from $T_1$ to $T_2$ is a binary relation $\R\subseteq\Sta_1\times\Sta_2$ such that for all states $s_1$ and $s_2$, $s_1\R s_2$ implies
\begin{enumerate}
    \item if $s_1\step{a}_1s_1'$, then there exist $s_2',s_2''\in\Sta_2$, such that $s_2\step{}_2^{*}s_2''\step{(a)}s_2'$, $s_1\R s_2''$ and $s_1'\R s_2'$;
    \item if $s_2\step{a}_2s_2'$, then there exist $s_1',s_1''\in\Sta_1$, such that $s_1\step{}_1^{*}s_1''\step{(a)}s_1'$, $s_1''\R s_2$ and $s_1'\R s_2'$.
\end{enumerate}
The transition systems $T_1$ and $T_2$ are \emph{branching bisimilar} (notation: $T_1\bbisim T_2$) if there exists a branching bisimulation $\R$ from $T_1$ to $T_2$ s.t. $\uparrow_1\R\uparrow_2$.

A branching bisimulation $\R$ from $T_1$ to $T_2$ is \emph{divergence-preserving} if, for all states $s_1$ and $s_2$, $s_1\R s_2$ implies
\begin{enumerate}
\setcounter{enumi}{2}
    \item if there exists an infinite sequence $(s_{1,i})_{i\in\mathbb{N}}$ such that $s_1=s_{1,0},\,s_{1,i}\step{\tau}s_{1,i+1}$ and $s_{1,i}\R s_2$ for all $i\in\mathbb{N}$, then there exists a state $s_2'$ such that $s_2\step{}^{+}s_2'$ and $s_{1,i}\R s_2'$ for some $i\in\mathbb{N}$; and
    \item if there exists an infinite sequence $(s_{2,i})_{i\in\mathbb{N}}$ such that $s_2=s_{2,0},\,s_{2,i}\step{\tau}s_{2,i+1}$ and $s_1\R s_{2,i}$ for all $i\in\mathbb{N}$, then there exists a state $s_1'$ such that $s_1\step{}^{+}s_1'$ and $s_1'\R s_{2,i}$ for some $i\in\mathbb{N}$.
\end{enumerate}
The transition systems $T_1$ and $T_2$ are \emph{divergence-preserving branching bisimilar} (notation: $T_1\bbisimd T_2$) if there exists a divergence-preserving branching bisimulation $\R$ from $T_1$ to $T_2$ such that $\uparrow_1\R\uparrow_2$.
\end{definition}

\section{Supplementary material for \RTMI}~\label{app:rtmi}
\paragraph*{Proof of Theorem~\ref{thm:RTM-infi}}
\begin{proof}
 Let $T=(\Sta_T,\step{}_T,\uparrow_T)$ be an $\Atau$-labelled countable transition system, and let $\encode{\_}: \Sta_T\rightarrow \mathbb{N}$ be an injective function encoding its states as natural numbers. Then, an RTM with infinite sets of action symbols and data symbols $\M(T)=(\Sta,\step{},\uparrow)$ is defined as follows.
\begin{enumerate}
    \item $\Sta=\{s,t,\uparrow\}$ is the set of control states.
    \item $\step{}$ is a finite $(\Dbox\times\A\times\Dbox\times\{L,R\})$-labelled \emph{transition relation}, and it consists of the following transitions:
        \begin{enumerate}
            \item $(\uparrow,\tau,\Box,\encode{\uparrow_T},R,s)$,
            \item $(s,\tau,\Box,\Box,L,t)$, and
            \item $(t,a,\encode{s_1},\encode{s_2},R,s)$ for every transition $s_1\step{a}_T s_2$.
        \end{enumerate}
    \item ${\uparrow}\in\Sta$ is the \emph{initial state}.
\end{enumerate}

Note that a transition step $s_1\step{a}s_2$ is simulated by a sequence
\begin{equation*}
(t,\tphd{\encode{s_1}}\Box)\step{a}(s,\encode{s_2}\tphd{\Box})\step{\tau}(t,\tphd{\encode{s_2}}\Box)
\enskip.
\end{equation*}
Then one can verify that $\T(\M(T))\bbisimd T$.
\end{proof}
\delete{
\paragraph*{Proof of Lemma~\ref{lemma:k-supported}}
\begin{proof}
\begin{enumerate}
\item $\Rightarrow$: We suppose that $K$ is a support of $\step{}_T$, then by the definition of support, for every $K$-automorphism $\pi_K$, we have $\pi_K(\step{}_T)=\step{}_T$. Therefore, for every $(s,a,t)\in{\step{}_T}$, we have $\pi_K(s,a,t)\in{\step{}_T}$.
\item $\Leftarrow$: We suppose that for every $(s,a,t)\in{\step{}_T}$ and for every $K$-automorphism $\pi_K$, we have $\pi_K(s,a,t)\in{\step{}_T}$ then we have $\pi_K(\step{}_T)=\step{}_T$, Therefore $K$ is a support of $\step{}_T$.
\end{enumerate}
\end{proof}

\paragraph*{Proof of Lemma~\ref{lemma:RTMA-transition-relation}}
\begin{proof}
As $\step{}_{\M}$ is a legal set. We take $K$ to be its the minimal support. Then we conclude the lemma from the definition of the support.
\end{proof}

\paragraph*{Proof of Lemma~\ref{lemma:illegal}}

\begin{proof}
Suppose that there exists such an \RTMA. If $x$ is in the support of the transitions of \RTMA, then the function can be trivially realized. However, introducing infinitely many atoms to the support will lead to an illegal transition relation. Therefore, $x$ is not in the support of the transitions of \RTMA.

As $x$ is not in the support, then for any transition that creates $x$, every $\pi(x)$ is created by another branch of execution, where $\pi$ is an arbitrary atom automorphism that preserves the support of the transitions of \RTMA. Hence, whenever an $x$-labelled transition is produced, some $\pi(x)$-labelled transition can be produced in some other branch of execution. Hence, this is not the required \RTMA.
\end{proof}
}
\section{Supplementary material for \RTMA}~\label{app:rtma}

\paragraph*{Details of Example~\ref{example:operation}}

We denote the current tape instance by $\bar{a}\bar{b}$, and we show the two ways to add a new atom to $\bar{b}$.

For the first case, we suppose that $x$ is the atom in $\bar{a}$ to be duplicated, and the first empty cell after $\bar{b}$ is the destination of the duplication. The machine could accomplish the task by the transitions $\cop\step{\tau[x/x]R}\cop_x\step{\tau[y/y]R}{}^{*}\step{\tau[\Box/x]}\finish$, which is realized by the following set of transitions.
\begin{eqnarray*}
&\{(\cop,\tau,x,x,R,\cop_x)\mid x\in\Atom\}\\
&\cup\{(\cop_x,\tau,y,y,R,\cop_x)\mid x,y\in\Atom\}\\
&\cup\{(\cop_x,\tau,\Box,x,R,\finish)\mid x\in\Atom\}
\enskip.
\end{eqnarray*}

This is a legal and orbit-finite set of transitions.

For the second case, the machine creates a fresh atom, by the following set of transitions,
\begin{eqnarray*}
&\{(\fresh,\tau,\Box,x,L,\chec_x)\mid x\in\Atom\}\\
&\cup\{(\chec_x,\tau,y,y,L,\chec_x)\mid x\neq y\wedge y\neq \Box\wedge x,y\in\Atom\}\\
&\cup\{(\chec_x,\tau,x,x,L,\refresh_x)\mid x\in\Atom\}\\
&\cup\{(\refresh_x,\tau,y,y,R,\refresh_x)\mid x\neq y\wedge x,y\in\Atom\}\\
&\cup\{(\refresh_x,\tau,x,y,L,\chec_y)\mid x\neq y\wedge x,y\in\Atom\}\\
&\cup\{(\chec_x,\tau,\Box,\Box,R,\finish)\mid x\in\Atom\}
\enskip.
\end{eqnarray*}

 The machine first creates an arbitrary atom $x$, and then it checks every atom on the tape whether it is identical with $x$. We suppose that $\Box$ indicates the end of the sequence of atoms on tape. If the check procedure succeeds, the creation is finished, otherwise, the machine creates another atom and checks again.
 We also verify that the above transitions form a legal and orbit-finite set.

\paragraph*{Details of Example~\ref{example:orb}}
We define an \RTMA{} $M=(\Sta_M,\step{}_M,\uparrow_M)$, and we show that $M$ suffices the requirement.

According to the assumption, we suppose that in the state $\start$, the tape instance is $\encode{\sete{x}}\tup{x}$. It suffices to show that within finitely many steps, the machine is able to write $x$ as one symbol on the tape.

Note that $X$ is an orbit-finite set, which means that there are finitely many distinct orbits that construct the set $X$. Therefore, there are finitely many distinct values of $\encode{\sete{x}}$ for all the elements in $X$. The machine associates with each value a program that calculates $\setb{\sete{x}}(\tup{x})$, which produces the elements from that specific orbit according to the valuation of $\bar{x}$. The machine enters the programme by entering the state $\encode{\sete{x}}$.

\begin{equation*}
\{(\start,\tau,\encode{\sete{x}},\encode{\sete{x}},R,\encode{\sete{x}})\mid x\in X\}
\enskip.
\end{equation*}

As $X$ is orbit-finite, there is an upper-bound for the length of the tuple $\bar{x}$ for every $x\in X$. The machine can represent a tuple as a data symbol.
We suppose that every element in $x$-orbit uses $n$ free variables in its structure, then $\bar{x}$ is a tuple of $n$ atoms. Now we consider the following set of transitions:
\begin{equation*}
\{(\encode{\sete{x}},\setb{\sete{x}}(\tup{x}),\tup{x},\setb{\sete{x}}(\tup{x}),R,\finish\mid a_1\ldots a_n\in \Atom\}
\enskip.
\end{equation*}
These transitions will create $x$ by $\setb{\sete{x}}(\tup{x})$. We suppose that $\tup{x}$ is of the form $(a_1,\ldots,a_n)$. The above set of transitions is orbit-finite, since there is an upper bound of $n$,

Moreover, we show that the machine is able to create a tuple $\tup{x}$ as one symbol, given that each atom in $\tup{x}$ is written on one tape cell, and ordered from left to right as the order of the atoms in the tuple.

The machine constructs $\bar{x}$ by duplicating the elements from each tape cell to the tuple one by one, using the transitions as follows:
\begin{eqnarray*}
&\{(\encode{\sete{x}},\tau,a,a,R,(\encode{\sete{x}},a))\mid a\in\Atom\}\\
&\cup\{((\encode{\sete{x}},a),\tau,b,b,R,(\encode{\sete{x}},a))\mid a,b\in\Atom\}\\
&\cup\{((\encode{\sete{x}},a),\tau,\Box,(a),L,(\encode{\sete{x}},a)')\mid a\in\Atom\}\\
&\cup\{((\encode{\sete{x}},a_i),\tau,(a_1,\ldots,a_{i-1}),(a_1,\ldots,a_{i-1},a_i),L,(\encode{\sete{x}},a_i)')\mid a_1,\ldots, a_i\in \Atom, 2 \leq i\leq n\}\\
&\cup\{((\encode{\sete{x}},a)',\tau,b,b,L,(\encode{\sete{x}},a)')\mid a,b\in \Atom,\, a\neq b\}\\
&\cup\{((\encode{\sete{x}},a)',\tau,a,a,R,\encode{\sete{x}})\mid a\in\Atom\}\\
\enskip.
\end{eqnarray*}
The machine first finds the atom to duplicate, and uses a state $(\encode{\sete{x}},a)$ to register the atom $a$. Then it moves the tape head to the tuple and adds the atom to that tuple. If the tuple is empty, then the machine produces $(a)$, otherwise, the machine adds $a$ to an existed tuple $(a_1,\ldots,a_{i-1})$, and enters the state $(\encode{\sete{x}},a)'$. Finally, the tape goes back to the atom it duplicated and enters state $\encode{\sete{x}}$ again to start the duplication of the next atom. This procedure ends by finishing the duplication of all the atoms in $\bar{x}$, and entering the state $\encode{\sete{x}}$ with $\tup{x}$ written under the tape head. Hence, the machine is ready to produce $\setb{\sete{x}}(\tup{x})$.

Since there is an upper bound of the length of the atom, this set of transitions is legal and orbit-finite. Moreover, there are finitely many orbits for the set $X$, which means that the machine needs finitely many such programs.
Hence, we have obtained an \RTMA{} $M$ that meets the requirement.

\paragraph*{Proof of Lemma~\ref{lemma:RTMA-effective}}
\begin{proof}

We let $T=(\Sta_T,\step{}_T,\uparrow_T)$ be an effective $\Atau$-labelled transition system with atoms, and let $K\subset\Atom$ be the minimal support of $T$. We show that there exists an \RTMA{} $\M=(\Sta_{\M},\step{}_{\M},\uparrow_{\M})$ such that $\T(\M)\bbisim T$.

As $T$ is effective, for every state $s\in\Sta_T$, the set $\mathalpha{out}(\encode{\sete{s}})=\{\encode{\sete{(s,a,t)}}\mid s\step{a}_T t\}$ is recursively enumerable. We use this fact to simulate the transition system. We describe the simulation in $3$ stages.

\begin{enumerate}
\item Initially, the tape is empty. Hence the initial configuration is $(\uparrow_{\M},\tphd{\Box})$. For simplicity, we do not denote the position of the tape head in the tape instances. The machine first writes the representation of the initial state $\uparrow_T$, i.e., $\encode{\sete{\uparrow_T}}\tup{\uparrow_T}$ on the tape, satisfying $\setb{\sete{\uparrow_T}}(\tup{\uparrow_T})=\uparrow_T$. As $T$ is legal, $\tup{\uparrow_T}$ consists of finitely many atoms. The initialization procedure is represented as follows,

\begin{equation*}
(\uparrow_{\M},\Box)\step{}^{*}(\enu,\encode{\sete{\uparrow_T}}\tup{\uparrow_T})
\enskip.
\end{equation*}

In the control state $\enu$, we assume that the tape instance is $\encode{\sete{s}}\tup{s}$, satisfying $\setb{\sete{s}}(\tup{s})=s$. As the transition system is effective, the machine is able to enumerate the structure of the outgoing transitions of $s$ as follows,

\begin{equation*}
(\enu,\encode{\sete{s}}\tup{s})\step{}^{*}(\gen,\encode{\sete{s}}\tup{s}\encode{\sete{(s,a,t)}})
\enskip.
\end{equation*}

\delete{
For state $s$, its subsequent transitions are $\{(s,a,t)\in \step{}_T\}$ and by the proposition of transition system with atoms, the set of transitions from $\pi_K(s)$ is $\pi_K{\{(s,a,t)\in \step{}_T\}}$. We use this fact to enumerate the transitions from $s$ and take a $\pi_K$ automorphism by the the atoms $\atom(\pi_K(s))$ which is already on the tape, and hence we obtained the transitions from $\pi_K(s)$.

Then the machine enumerates $\mathalpha{out}(s)$,  and writes $(\encode{a},\encode{t})$ (which are represented by $\encode{\orb_{[a]}},\encode{\atom(a)},\encode{\orb_{[t]}},\encode{\atom(t)}$ respectively) on the tape whenever a transition $(s,a,t)$ is enumerated.

The step of enumeration is represented as follows,

\begin{eqnarray*}
&(\enu,\encode{\orb_{[s]}}\encode{\atom(s)}\atom(\pi_K(s)))\step{}^{*}\\
&(\gen,\encode{\orb_{[s]}}\encode{\atom(s)}\encode{\orb_{[a]}}\encode{\atom(a)}\encode{\orb_{[t]}}\encode{\atom(t)}\atom(\pi_K(s)))
\enskip.
\end{eqnarray*}
}
\item In the second stage, the \RTMA{} produces the tuples of atoms $\tup{a}$ and $\tup{t}$ that valuates the free variables of $a$ and $t$. The valuation creates fresh atoms when necessary and preserves atoms from $K$ and $\tup{s}$.

     We denote the tuple of free variables of $\sete{(s,a,t)}$ by $\bar{x}$, and tuples of free variables of $\sete{s}$, $\sete{a}$ and $\sete{t}$ by $\bar{x_s}$, $\bar{x_a}$ and $\bar{x_t}$ respectively. Note that all the variables in $\bar{x_s}$, $\bar{x_a}$ and $\bar{x_t}$ are also in $\bar{x}$. Since $\encode{\sete{(s,a,t)}}$ is already on the tape, the following terms are computable:
    \begin{enumerate}
        \item the set builder expressions of $a$ and $t$: $\encode{\sete{a}}$ and $\encode{\sete{t}}$;
        \item the tuples of free variables: $\bar{x}$, $\bar{x_s}$, $\bar{x_a}$ and $\bar{x_t}$;
    \end{enumerate}
    We show the above statements as follows. We define a triple $(s,a,t)$ by $\{\{s\},\{s,a\},\{s,a,t\}\}$. We use the standard G\"{o}del numbering on sets and variables. The encoding $\encode{\sete{(s,a,t)}}$ is equal to $\encode{\{\{\sete{s}\},\{\sete{s},\sete{a}\},\{\sete{s},\sete{a},\sete{t}\}\}}$. $\encode{\sete{a}}$ and $\encode{\sete{t}}$ are all computable because the projection operation is computable in G\"{o}del encoding of ordered triples. Moreover, the projection from $\encode{\sete{x}}$ to the free variables used in each elements are computable.

   Then we evaluate $\bar{x}$ to tuples of atoms. $\tup{s}$ is a valuation of $\bar{x_s}$. We evaluate $\bar{x}$ by distinguishing two cases:
  \begin{enumerate}
        \item if a variable $y$ in $\bar{x}$ also appears in $\bar{x_s}$, then we duplicate the valuation of that variable from $\tup{s}$ to valuate $y$;
        \item otherwise, we create a fresh atom to valuate $y$.
    \end{enumerate}
    By Example~\ref{example:operation}, the above two operations are valid by \RTMA s.
    Since $\bar{x_a}$ and $\bar{x_t}$ are both sub-tuples of $\bar{x}$, the machine duplicates the valuation from $\bar{x}$ to create $\tup{a}$ and $\tup{t}$.

    Hence, we get $\encode{\sete{a}}$, $\encode{\sete{t}}$, $\tup{a}$ and $\tup{t}$ satisfying that $\setb{\sete{a}}(\tup{a})=a'$, $\setb{\sete{t}}(\tup{t})=t'$, and there exists an $K\cup\tup{s}$-automorphism $\pi$ which preserves all the atoms in $K\cup\tup{s}$, such that $\pi(s,a',t')=(s,a,t)$. By the property of transition system with atoms, $(s,a',t')\in\step{}_T$. Moreover, by Example~\ref{example:operation}, during the generation of fresh variables, every fresh variable from the universe of $\Atom$ can be generated, the \RTMA{} is able to create all the transitions which are equivalent to $(s,a,t)$ up to $K\cup\tup{s}$-automorphism. We denote this stage as follows:

\begin{equation*}
(\gen,\encode{\sete{s}}\tup{s}\encode{\sete{(s,a,t)}})\step{}^{*}(\act,\encode{\sete{s}}\encode{\sete{a}}\encode{\sete{t}}\tup{s}\tup{a}\tup{t})
\enskip.
\end{equation*}
\delete{
 We use $\pi_{K\cup s}$ to denote some arbitrary atom automorphism that preserves the support $K$ and the atoms from the state $\pi_K(s)$, and $\pi_{K\cup s\cup a}$ to denote some arbitrary atom automorphism that preserves $K$, $\pi_K(s)$ and $\atom(\pi_{K\cup s}(a))$. Then the machine nondeterministically produces $\atom(\pi_{K\cup s}(a))$ and $\atom(\pi_{K\cup s\cup a}(t))$ according to $\atom(s)$, $\encode{\atom(s)}$, $\encode{\atom(a)}$ and $\encode{\atom(t)}$, for some arbitrary $\pi_{K\cup s}$ and $\pi_{K\cup s\cup a}$.

The machine first generates $\atom(\pi_{K\cup s}(a))$. For every element $\encode{x}\in\encode{\atom(a)}$, we distinguish with two cases:

\begin{enumerate}
\item if $\encode{x}$  also in $\encode{\atom(s)}$, then the machine find the the corresponding element from $\atom(\pi_K(s))$ and duplicate that one to fills the position of $\encode{x}$ in $\atom(\pi_{K\cup s}(a))$;
\item if $\encode{x}$ not in $\encode{\atom(s)}$, then the machine create a fresh atom and fills the position of $\encode{x}$ in $\atom(\pi_{K\cup s}(a))$.
\end{enumerate}

Then the machine then generates $\atom(\pi_{K\cup s\cup a}(t))$. For every element $\encode{x}\in\encode{\atom(t)}$, we also distinguish with two cases:
\begin{enumerate}
\item if $\encode{x}$ is also in $\encode{\atom(s)}$ or $\encode{\atom(a)}$ , then the machine find the the corresponding element from $\atom(\pi_K(s))$ or $\atom(\pi_{K\cup s}(a))$ and duplicate that one to fills the position of $\encode{x}$ in $\atom(\pi_{K\cup s\cup a }(t))$;
\item if $\encode{x}$ is not in $\encode{\atom(s)}$ nor in $\encode{\atom(a)}$ , then the machine create a fresh atom and fills the position of $\encode{x}$ in $\atom(\pi_{K\cup s\cup a}(t))$.
\end{enumerate}

By Example~\ref{example:operation}, duplication and creating of atoms are valid operations for \RTMA.

We denote this step as follows,
\begin{eqnarray*}
&(\gen,\encode{\orb_{[s]}}\encode{\orb_{[a]}}\encode{\orb_{[t]}}\encode{\atom(s)}\encode{\atom(a)}\encode{\atom(t)}\atom(\pi_K(s)))\step{}^{*}\\
&(\act,\encode{\orb_{[s]}}\encode{\orb_{[a]}}\encode{\orb_{[t]}}\encode{\atom(s)}\encode{\atom(a)}\encode{\atom(t)}\atom(\pi_K(s))\atom(\pi_{K\cup s}(a))\atom(\pi_{K\cup s\cup a}(t)))
\enskip.
\end{eqnarray*}
}

\item In the third stage, the \RTMA{} generates the action label $a'$ and chooses to execute the transition or to continue the enumeration.

As $\Atau$ is a legal and orbit-finite set, by Example~\ref{example:orb}, the \RTMA{} is able to create a symbol $a'$ such that $\setb{\sete{a}}\tup{a}=a'$.

\begin{equation*}
(\act,\encode{\sete{s}}\encode{\sete{a}}\encode{\sete{t}}\tup{s}\tup{a}\tup{t})\step{}^{*}(\tran,\encode{\sete{s}}\encode{\sete{a}}\encode{\sete{t}}\tup{s}\tup{a}\tup{t}a')
\enskip.
\end{equation*}

Then the \RTMA{} has two choices: executing the $a'$-labelled transition and starting the next round of simulation or returning to the first stage.
\begin{eqnarray*}
&(\tran,\encode{\sete{s}}\encode{\sete{a}}\encode{\sete{t}}\tup{s}\tup{a}\tup{t}a')\step{a'}\step{}^{*}(\enu,\encode{\sete{t}}\tup{t})\\
&(\tran,\encode{\sete{s}}\encode{\sete{a}}\encode{\sete{t}}\tup{s}\tup{a}\tup{t}a')\step{}^{*}(\enu,\encode{\sete{s}}\tup{s})\\
\end{eqnarray*}
\delete{
Hence, the machine will produce a label $\pi_{K\cup s}(a)$ from $\atom(\pi_{K\cup s}(a))$ and $\encode{\orb_{[a]}}$. By Corollary~\ref{cor:orb}, this step is computable by an \RTMA. According to the property of transition systems with atoms, the transition to be simulated $(\pi_K(s),\pi_{K\cup s}(a),\pi_{K\cup s\cup a}(t))\in\step{}_T$. Moreover, the procedure of generating fresh atoms, it is guaranteed that every transition $(\pi_K(s),a',t')\in\step{}_T$ satisfying that there exists some $\pi_{K\cup s}$ such that $\pi_{K\cup s}(\pi_K(s),a,t)=(\pi_K(s),a',t')$ are produced in the previous step.

This procedure of producing the action label is represented as follows,

\begin{eqnarray*}
&(\act,\encode{\orb_{[s]}}\encode{\orb_{[a]}}\encode{\orb_{[t]}}\encode{\atom(s)}\encode{\atom(a)}\encode{\atom(t)}\\
&\atom(\pi_K(s))\atom(\pi_{K\cup s}(a))\atom(\pi_{K\cup s\cup a}(t)))\step{}^{*}\\
&(\tran,\encode{\orb_{[s]}}\encode{\orb_{[a]}}\encode{\orb_{[t]}}\encode{\atom(s)}\encode{\atom(a)}\encode{\atom(t)}\\
&\atom(\pi_K(s))\atom(\pi_{K\cup s}(a))\atom(\pi_{K\cup s\cup a}(t))\pi_{K\cup s}(a))
\enskip.
\end{eqnarray*}

\item $\tran\Rightarrow\enu$:

Finally the machine simulates the transition or continues the enumeration.

\begin{eqnarray*}
&(\tran,\encode{\orb_{[s]}}\encode{\orb_{[a]}}\encode{\orb_{[t]}}\encode{\atom(s)}\encode{\atom(a)}\encode{\atom(t)}\\
&\atom(\pi_K(s))\atom(\pi_{K\cup s}(a))\atom(\pi_{K\cup s\cup a}(t))\pi_{K\cup s}(a))\step{\pi_{K\cup s}(a)}\step{}^{*}\\
&(\enu,\encode{\orb_{[t]}}\encode{\atom(t)}\atom(\pi_{K\cup s\cup a}(t)))\\
&(\tran,\encode{\orb_{[s]}}\encode{\orb_{[a]}}\encode{\orb_{[t]}}\encode{\atom(s)}\encode{\atom(a)}\encode{\atom(t)}\\
&\atom(\pi_K(s))\atom(\pi_{K\cup s}(a))\atom(\pi_{K\cup s\cup a}(t))\pi_{K\cup s}(a))
\step{}^{*}\\
&(\enu,\encode{\orb_{[s]}}\encode{\atom(s)}\atom(\pi_K(s)))
\enskip.
\end{eqnarray*}}
\end{enumerate}
We can verify that before the $a'$-labelled transition, the transition system of the machine preserves its states modulo $\bbisim$ by a sequence of $\tau$-transitions which leads back to the configuration $(\enu,\encode{\sete{s}}\tup{s})$. Moreover, from the above analysis, we have $(s,a',t')\in\step{}_T$; and every transition obtained by an $K\cup\tup{s}$-automorphism from $(s,a,t)$ can be simulated by the \RTMA{} $M$. Therefore, we conclude that $T\bbisim \T(\M)$.
\end{proof}

\paragraph*{Proof of Lemma~\ref{lemma:RTMA-effective2}}
\begin{proof}
It is obvious that $\T(\M)$ is effective.

Let $\M=(\Sta_{\M},\step{}_{\M},\uparrow_{\M})$, then there exists a finite set of atoms $K\subset\Atom$ such that, for every $(s,a,d,e,M,t)\in{\step{}_{\M}}$, and for every $K\mbox{-automorphism}\,\pi_K$, we have $\pi_K(s,a,d,e,M,t)\in{\step{}_{\M}}$. It follows that the transition system $\T(\M)$ is legal.
\end{proof}

\section{The $\pi$-calculus}~\label{app:pi}

\paragraph*{Transition systems of the $\pi$-calculus}
The $\pi$-calculus was proposed by Milner, Parrow and Walker~\cite{Milner1992} as a language to specify processes with link mobility. In this paper, we shall consider the version presented in the textbook by Sangiorgi and Walker~\cite{SW01}, excluding the match prefix.

We presuppose a countably infinite set $\N$ of names; we use strings of lower case letters for elements of $\N$.
The \emph{prefixes}, \emph{processes} and \emph{summations} of the $\pi$-calculus are, respectively, defined by the following grammar:
\begin{align*}
\pi\      & \coloneqq\ \outcap{x}{y}\ \mid\ \incap{x}{z}\ \mid\ \taucap \qquad (x,y,z\in \N)\\
P\    & \coloneqq\ M\ \mid\  P\parc P\ \mid\ \restr{z}{P}\ \mid\ \repl{P}\\
M\   & \coloneqq\ \nil\ \mid\ \pref{\pi}P \mid\ M \altc M\enskip.
\end{align*}

We use $P\{z/y\}$ to denote a $\pi$-term obtained by substituting every occurrence of $y$ to $z$ in $P$.

An $\alpha$-conversion between $\pi$-terms is defined in~\cite{SW01} as a finite number of renaming of bound names. We write $P\aeq Q$ if $P$ and $Q$ are two $\pi$-terms that are $\alpha$-convertible.

We define the operational behaviour of $\pi$-terms by means of the structural operational semantics in Fig.~\ref{fig:pi-semantics}, in which $\piact{}$ ranges over the set of actions of the $\pi$-calculus.

The transition system associated with a $\pi$-term is defined as follows:

\begin{definition}~\label{lts-piterm}
Let $P$ be a $\pi$-term. $\T(P)=(\Sta_{P},\step{}_{P},\uparrow_{P})$ is the transition system associated with $P$, where
\begin{enumerate}
    \item $\Sta_{P}$ is the set of $\alpha$-equivalence classes of all reachable $\pi$-terms from $P$ by the operational semantics;
    \item $\step{}_{P}$ is the set of transitions between $\alpha$-equivalence classes of all reachable $\pi$-terms; and
    \item $\uparrow_{P}$ is the $\alpha$-equivalence class of $P$.
\end{enumerate}
\end{definition}

\begin{figure}
\begin{center}
\fbox{
\begin{minipage}[t]{0.9\textwidth}
%$\mathrm{STRUCT}\quad\inference{P'\equiv P,\,P\step{\pi}Q,\,Q'\equiv Q}{P'\step{\piact{}}Q'}$\\
$\mathrm{PREFIX}\quad\inference{\,}{\tau.P\step{\tau}P}\quad \inference{}{\overline{x}y.P\step{\overline{x}y}P}\quad\inference{}{x(y).P\step{xz}P\{z/y\}}$\\
$\mathrm{SUM_L}\quad\inference{P\step{\piact{}}P'}{P+Q\step{\piact{}}P'}%\quad\inference{Q\step{\piact{}}Q'}{(P+Q)\step{\piact{}}Q'}$\\
%%$\mathrm{MATCH}\quad\inference{P\step{\alpha}P'}{if\,x=x\,then\,P\step{\alpha}P'}$\\
%%$\mathrm{MISMATCH}\quad\inference{P\step{\alpha}P',\,x\neq y}{if\,x\neq y\,then\,P\step{\alpha}P'}$\\
\quad\mathrm{PAR_L}\quad\inference{P\step{\piact{}}P'}{P\parc{Q}\step{\piact{}}P'\parc{Q}}\,\bn{\piact{}}\cap \fn{Q}=\emptyset$\\
$\mathrm{COM_L}\quad\inference{P\step{\overline{x}y}P',\,Q\step{xy}Q'}{P\parc{Q}\step{\tau}P'\parc{Q'}}\quad
\mathrm{CLOSE_L}\quad\inference{P\step{\overline{x}(z)}P',\,Q\step{xz}Q'}{P\parc{Q}\step{\tau}\restr{z}{(P'\parc{Q'})}}\,z\notin \fn{Q}$\\
$\mathrm{RES}\quad\inference{P\step{\piact{}}P'}{\restr{z}{P}\step{\piact{}}\restr{z}{P'}}\,z\notin\piact{}\quad
\mathrm{OPEN}\quad\inference{P\step{\overline{x}z}P'}{(z)P\step{\overline{x}(z)}P'}\,z\neq x$\\
$\mathrm{REP}
%  \quad\inference{P\mid !P\step{\piact{}}P'}{!P\step{\piact{}}P'}
  \quad\inference{P\step{\piact{}}P'}{\repl{P}\step{\piact{}}P'\parc\repl{P}}
  \quad\inference{P\step{\outact{x}{y}}P',\,P\step{\inact{x}{y}}P''}{\repl{P}\step{\tauact{}}(P'\parc P'')\parc\repl{P}}
  \quad\inference{P\step{\boutact{x}{z}}P',\,P\step{\inact{x}{z}}P''}{\repl{P}\step{\tauact}\restr{z}{(P'\parc P'')}\parc\repl{P}}$\\
$\mathrm{ALPHA}\quad\inference{P\step{\piact{}}P'}{Q\step{\piact{}}P'}\,Q\aeq P$.
\end{minipage}
}\end{center}
\caption{Operational rules for the $\pi$-calculus}\label{fig:pi-semantics}
\end{figure}

\paragraph*{Proof of Corollary~\ref{cor:rtma-pi}}

\begin{proof}
We use $\N$ to denote the countable set of names used in the $\pi$-calculus, and we suppose that $\N=\Atom$. We suppose that $P$ is an arbitrary $\pi$-calculus process. We use $\fn{P}$ to denote the set of free names involved in $P$ and $\bn{P}$ to denote the set of bound names involved in $P$.

By Theorem~\ref{thm:RTMA-effective}, it is sufficient to show that for every $\pi$-calculus process $P$, the transition system $\T(P)$ is an effective legal transition system with atoms. The transition system is effective by the effectiveness of structural operational semantics of the $\pi$-calculus. Therefore, by Definition~\ref{def:ltsa}, it is sufficient to show that there is a finite set $K\subset\N$ such that $K$ is a support of $\T(P)$. We let $\T(P)=(\Sta_{\pi},\step{}_{\pi},P)$ be an $\Atau$ labelled transition system, and we take $K=\fn{P}\cup\{\tau\}$. Note that $\Sta_P$ is the set of $\pi$-terms and $\Atau$ is the set of labels, and hence $\Sta_{\pi}$ is a set with atoms and $\Atau$ is an orbit-finite set with atoms. To show that $K$ is a support of $\T(P)$, we only need to show that for every $(s,a,t)\in\step{}_{\pi}$, and for every $K$-automorphism $\pi_K$, $\pi_K(s,a,t)\in\step{}_{\pi}$.

We let $\pi_K$ be an arbitrary $K$-automorphism, and we show that $\pi_K(P)$ is an $\alpha$-conversion of $P$, i.e., $\pi_K(P)\aeq P$. We show it by a structural induction on $P$.
In the base case, $P=\nil$, then it is trivial that $\pi_K(P)=\nil\aeq P$.

For the step case, we distinguish with 7 cases.
\begin{enumerate}
\item If $P=\pref{\tau}P'$, by induction hypothesis, we have $\pi_K(P')\aeq P'$. Hence, we have $\pi_K(\pref{\tau}P')=\pref{\tau}\pi_K(P')\aeq \pref{\tau}P'=P$.
\item If $P=\pref{\outcap{x}{y}}P'$, by induction hypothesis, we have $\pi_K(P')\aeq P'$. Moreover, $x,y\in\fn{P}\in K$. Hence, we have $\pi_K(\pref{\outcap{x}{y}}P')=\pref{\outcap{x}{y}}\pi_K(P')\aeq \pref{\outcap{x}{y}}P'=P$.
\item If $P=\pref{\incap{x}{y}}P'$, by induction hypothesis, we have $\pi_K(P')\aeq P'$. Moreover, $x\in\fn{P}\in K$. Hence, we have $\pi_K(\pref{\incap{x}{y}}P')=\pref{\incap{x}{\pi_K(y)}}\pi_K(P')\aeq\pref{\incap{x}{y}}P'=P$.
\item If $P=P_1+P_2$, by induction hypothesis, we have $\pi_K(P_1)\aeq P_1,\,\pi_K(P_2)\aeq P_2$. Hence, we have $\pi_K(P_1+P_2)=\pi_K(P_1)+\pi_K(P_2)\aeq P_1+P_2=P$.
\item If $P=P_1\parc P_2$, by induction hypothesis, we have $\pi_K(P_1)\aeq P_1,\,\pi_K(P_2)\aeq P_2$. Hence, we have $\pi_K(P_1\parc P_2)=\pi_K(P_1)\parc \pi_K(P_2)\aeq P_1\parc P_2=P$.
\item If $P=\restr{z}{P'}$, by induction hypothesis, we have $\pi_K(P')\aeq P'$. Moreover, $\pi_K(z)=z'\notin K$. Hence, we have $\pi_K(\restr{z}{P'})=\restr{z'}{\pi_K(P')}\aeq \restr{z}{P'}=P$.
\item If $P=\repl{P'}$, by induction hypothesis, we have $\pi_K(P')\aeq P$. Hence, we have $\pi_K(\repl{P'})=\repl{\pi_K(P')}\aeq \repl{P'}=P$.
\end{enumerate}

By structural induction, we have shown that for every $K$-automorphism $\pi_K$, $\pi_K(P)\aeq P$. Therefore, we have $\pi_K(P)\in\Sta_{\pi}$. Hence, $K$ is a support of $\Sta_{\pi}$.

Next we still let $\pi_K$ be an arbitrary $K$-automorphism, and we show that for every transition $P\step{\piact}_{\pi}Q\in\step{}_{\pi}$, it satisfies that $\pi_K(P\step{\piact}_{\pi}Q)\in\step{}_{\pi}$ by an induction on the structural operational semantics of the $\pi$ calculus.

We construct a proof tree according to the structural operational semantics in Figure~\ref{fig:pi-semantics} for every transition $(P,\piact,Q)\in\step{}_{\pi}$. The induction hypothesis is that, if $(P,\piact,Q)$ is induced from a set of transitions $\mathit{Pre}(P,\piact,Q)\subset\step{}_{\pi}$ , then for every transition $(P_i,\piact_i,Q_i)\in\mathit{Pre}(P,\piact,Q)$, there is $\pi_K(P_i,\piact_i,Q_i)\in\step{}_{\pi}$.
For the base case, the $\nil$ process cannot do any transition, then the property trivially holds. For the step case, we distinguish with several cases as follows.

\begin{enumerate}
\item If the transition is $\pref{\tauact}P\step{\tauact}_{\pi}P$, then we have $\pi_K(\pref{\tauact}P)\step{\pi_K(\tauact)}_{\pi}\pi_K(P)$.
\item If the transition is $\pref{\outcap{x}{y}}P\step{\outact{x}{y}}_{\pi}P$, then we have $\pi_K(\pref{\outcap{x}{y}}P)\step{\pi_K(\outact{x}{y})}_{\pi}\pi_K(P)$.
\item If the transition is $\pref{\incap{x}{y}}P\step{\outact{x}{z}}_{\pi}P\{z/y\}$, then we have $\pi_K(\pref{\incap{x}{y}}P)\step{\pi_K(\inact{x}{z})}_{\pi}\pi_K(P\{z/y\})$.
\item If the transition is $P_1+P_2\step{\piact}_{\pi}Q$, then there are two cases
\begin{enumerate}
    \item if $P_1\step{\piact}_{\pi}Q$. By induction hypothesis, we have $\pi_K(P_1)\step{\pi_K(\piact)}_{\pi}\pi_K(Q)$. Hence, we have $\pi_K(P_1+P_2)\step{\pi_K(\piact)}_{\pi}\pi_K(Q)$.
    \item If $P_2\step{\piact}_{\pi}Q$. The proof is symmetric with the previous case.
\end{enumerate}
\item If the transition is $P_1\parc P_2\step{\piact}_{\pi}Q_1\parc P_2$, then we have $P_1\step{\piact}_{\pi}Q_1$ and $\bn{\piact}\cap\fn{P_2}=\emptyset$. By induction hypothesis, we have $\pi_K(P_1)\step{\pi_K(\piact)}_{\pi}\pi_K(Q_1)$. Moreover, there is $\bn{\pi_K(\piact)}\cap\fn{\pi_K(P_2)}=\emptyset$. Hence we have $\pi_K(P_1\parc P_2)\step{\pi_K(\piact)}_{\pi}\pi_K(Q_1\parc P_2)$.
\item If the transition is $P_1\parc P_2\step{\piact}_{\pi}P_1\parc Q_2$. The  proof is symmetric with the previous case.
\item If the transition is $P_1\parc P_2\step{\tauact}_{\pi}Q_1\parc Q_2$, then we have $P_1\step{\outact{x}{y}}_{\pi}Q_1,\,P_2\step{\inact{x}{y}}_{\pi}Q_2$ (or the symmetrical case). By induction hypothesis, we have $\pi_K(P_1)\step{\pi_K(\outact{x}{y})}_{\pi}\pi_K(Q_1)$ and $\pi_K(P_2)\step{\pi_K(\inact{x}{y})}_{\pi}\pi_K(Q_2)$. Hence we have $\pi_K(P_1\parc P_2)\step{\pi_K(\tauact)}_{\pi}\pi_K(Q_1\parc Q_2)$.
\item If the transition is $P_1\parc P_2\step{\tauact}_{\pi}\restr{z}{(Q_1\parc Q_2)}$, then we have $P_1\step{\boutact{x}{z}}_{\pi}Q_1$, $P_2\step{\inact{x}{z}}_{\pi}Q_2$, and $z\notin\fn{P_2}$ (or the symmetrical case). By induction hypothesis, we have $\pi_K(P_1)\step{\pi_K(\boutact{x}{z})}_{\pi}\pi_K(Q_1)$ and $\pi_K(P_2)\step{\pi_K(\inact{x}{z})}_{\pi}\pi_K(Q_2)$. Moreover, there is $\pi_K(z)\notin\fn{\pi_K(P_2)}$. Hence, we have $\pi_K(P_1\parc P_2)\step{\pi_K(\tauact)}_{\pi}\pi_K(\restr{z}{(Q_1\parc Q_2)})$.

\item If the transition is $\restr{z}{P}\step{\piact}_{\pi}\restr{z}{Q}$, then we have $P\step{\piact}_{\pi} Q$ and $z\notin \piact$. By induction hypothesis, we have $\pi_K(P)\step{\pi_K(\piact)}_{\pi}\pi_K(Q)$. Moreover, there is $\pi_K(z)\notin\pi_K(\piact)$. Hence, we have $\pi_K(\restr{z}{P})\step{\pi_K(\piact)}_{\pi}\pi_K(\restr{z}{Q})$.
\item If the transition is $\restr{z}{P}\step{\boutact{x}{z}}_{\pi}Q$, then we have $P\step{\outact{x}{z}}Q$, and $z\neq x$.  By induction hypothesis, we have $\pi_K(P)\step{\pi_K(\outact{x}{z})}_{\pi}\pi_K(Q)$. Moreover, there is $\pi_K(z)\neq\pi_K(x)$. Hence, we have $\pi_K(\restr{z}{P})\step{\pi_K(\boutact{x}{z})}_{\pi}\pi_K(Q)$.
\item If the transition is $\repl{P}\step{\piact}_{\pi}Q\parc \repl{P}$, then we have $P\step{\piact}_{\pi}Q$. By induction hypothesis, we have $\pi_K(\repl{P})\step{\pi_K(\piact)}_{\pi}\pi_K(Q)$. Hence we have  $\pi_K(P)\step{\pi_K(\piact)}_{\pi}\pi_K(Q\parc \repl{P})$.
\item If the transition is $\repl{P}\step{\tauact}_{\pi}(Q_1\parc Q_2)\parc \repl{P}$, then we have $P\step{\outact{x}{y}}_{\pi}Q_1$, and $P\step{\inact{x}{y}}_{\pi}Q_2$. By induction hypothesis, we have $\pi_K(P)\step{\pi_K(\outact{x}{y})}_{\pi}\pi_K(Q_1)$, and $\pi_K(P)\step{\pi_K(\inact{x}{y})}_{\pi}\pi_K(Q_2)$. Hence we have  $\pi_K(P)\step{\pi_K(\tauact)}_{\pi}\pi_K((Q_1\parc Q_2)\parc \repl{P})$.
\item If the transition is $\repl{P}\step{\tauact}_{\pi}\restr{z}{(Q_1\parc Q_2)}\parc \repl{P}$, then we have $P\step{\boutact{x}{z}}_{\pi}Q_1$, and $P\step{\inact{x}{z}}_{\pi}Q_2$. By induction hypothesis, we have $\pi_K(P)\step{\pi_K(\boutact{x}{z})}_{\pi}\pi_K(Q_1)$, and $\pi_K(P)\step{\pi_K(\inact{x}{z})}_{\pi}\pi_K(Q_2)$. Hence we have  $\pi_K(P)\step{\pi_K(\tauact)}_{\pi}\pi_K(\restr{z}{(Q_1\parc Q_2)}\parc \repl{P})$.

\item If the transition is $P\step{\piact}_{\pi} Q$, and $P\aeq P_1$, then $P_1\step{\piact}_{\pi} Q$. By induction hypothesis $\pi_K(P_1)\step{\pi_K(\piact)}_{\pi} \pi_K(Q)$. Moreover, using the statement $\pi_K(P)\aeq P$, we have $\pi_K(P)\aeq P\aeq P_1\aeq\pi_K(P_1)$. Hence, we have $\pi_K(P)\step{\pi_K(\piact)}_{\pi}\pi_K(Q)$.
\end{enumerate}

Using the induction on the depth of the proof tree of the transition, we have shown that for every transition $P\step{\piact}_{\pi}Q$, we have $\pi_K(P)\step{\pi_K(\piact)}_{\pi}\pi_K(Q)$. We conclude that $K$ is a support of $\step{}_{\pi}$. Therefore, the transition system $\T(P)$ is an effective legal transition system with atoms.

As a consequence of Theorem~\ref{thm:RTMA-effective}, for every $\pi$-calculus process $P$, the transition system $\T(P)$ is nominally executable.
\end{proof}

\section{mCRL2}~\label{app:mcrl2}
\paragraph*{Proof of Corollary~\ref{cor:mCRL2}}
\begin{proof}
Consider the following mCRL2 specification:
\begin{eqnarray*}
&&\mathit{act\, num:\, Nat};\\
&&\mathit{init\,sum\, v:\, Nat\, .\, num(2 * v)};
\end{eqnarray*}

It defines a transition system that includes a set of transitions from the initial state labelled by all even natural numbers as follows:
\begin{equation*}
\{(\uparrow,2n,\downarrow)\mid n\in\mathbb{N}\}
\enskip.
\end{equation*}

This transition system does not have a finite support, therefore, it is not nominally executable.
\end{proof}